\documentclass[sigconf,screen,nonacm]{acmart}

\usepackage{amsmath}
\usepackage{amsfonts}
\usepackage{amsthm}
\usepackage{isomath}
\usepackage{graphicx}
\usepackage{color}
\usepackage{tikz}

\usetikzlibrary{positioning,automata}

\newtheorem{definition}{Definition}[section]
\newtheorem*{definition*}{Definition}
\newtheorem{proposition}{Proposition}[section]
\newtheorem{example}{Example}[section]

\newtheorem{theorem}{Theorem}
\newtheorem{lemma}{Lemma}
\newtheorem*{lemma*}{Lemma}

\newtheorem{corollary}{Corollary}[section]

\newcommand{\bra}[1]{\langle #1 \vert}
\newcommand{\ket}[1]{\vert #1 \rangle}

\newcommand{\HH}{\mathcal{H}}
\newcommand{\DD}{\mathcal{D}}

\newcommand{\EE}{\mathcal{E}}

\newcommand{\sem}[1]{[\![ #1 ]\!]}

\newcommand{\NN}{\mathbb{N}}
\newcommand{\RR}{\mathbb{R}}
\newcommand{\CC}{\mathbb{C}}
\newcommand{\ZZ}{\mathbb{Z}}

\newcommand{\guard}{\Box}

\newcommand{\abs}[1]{| #1 |}
\newcommand{\norm}[1]{\lVert #1 \rVert}

\newcommand{\qprogs}{ \textbf{qProgs} }
\newcommand{\qvar}{\textbf{qVar}}

\DeclareMathOperator{\tr}{tr}

\DeclareMathOperator{\ert}{ert}
\DeclareMathOperator{\ERT}{ERT}
\DeclareMathOperator{\Un}{Un}
\DeclareMathOperator{\In}{In}
\DeclareMathOperator{\Sk}{Sk}
\DeclareMathOperator{\Br}{Br}
\DeclareMathOperator{\Wh}{Wh}

\DeclareMathOperator{\diag}{diag}
\newcommand{\skp}{\mathbf{skip}}
\newcommand{\ifb}{\mathbf{if}}
\newcommand{\ife}{\mathbf{fi}}
\newcommand{\while}{\mathbf{while}}
\newcommand{\wdo}{\mathbf{do}}
\newcommand{\wod}{\mathbf{od}}

\newcommand{\oq}{\overline{q}}
\newcommand{\zeromat}{\mathbf{0}}

\copyrightyear{2022}
\acmYear{2022}
\setcopyright{acmcopyright}\acmConference[LICS '22]{37th Annual ACM/IEEE Symposium on Logic in Computer Science (LICS)}{August 2--5, 2022}{Haifa, Israel}
\acmBooktitle{37th Annual ACM/IEEE Symposium on Logic in Computer Science (LICS) (LICS '22), August 2--5, 2022, Haifa, Israel}
\acmPrice{15.00}
\acmDOI{10.1145/3531130.3533327}
\acmISBN{978-1-4503-9351-5/22/08}

\begin{document}

\title{Quantum Weakest Preconditions for  Reasoning about  Expected Runtimes  of Quantum Programs (Extended Version)}
   
  \author{Junyi Liu}
  \affiliation{
  \institution{Institute of Software, CAS}
  \city{Beijing}
  \country{China}
  }
  \affiliation{
  \institution{University of Chinese Academy of Sciences}
  \city{Beijing}
  \country{China}
  }
  \email{liujy@ios.ac.cn}
  
  \author{Li Zhou}
  \affiliation{
  \institution{Max Planck Institute for Security and Privacy (MPI-SP)}
  \city{Bochum}
  \country{Germany}
  }
  \email{zhou31416@gmail.com}
  
  \author{Gilles Barthe}
  \affiliation{
  \institution{Max Planck Institute for Security and Privacy (MPI-SP)}
  \city{Bochum}
  \country{Germany}
  }
  \affiliation{
  \institution{IMDEA Software Institute}
  \city{Madrid}
  \country{Spain}
  }
  \email{gjbarthe@gmail.com}
  
  \author{Mingsheng Ying}
  \affiliation{
  \institution{Institute of Software, CAS}
  \city{Beijing}
  \country{China}
  }
  \affiliation{
  \institution{Tsinghua University}
  \city{Beijing}
  \country{China}
  }
  \email{yingms@ios.ac.cn}

 \begin{abstract} 
   We study \textit{expected runtimes} for quantum programs. Inspired by recent work on probabilistic programs, we first define \textit{expected runtime} as a generalisation of \textit{quantum weakest precondition}. Then, we show that 
   the expected runtime of a quantum program can be represented as the expectation of an \textit{observable} (in physics). A method for computing the expected runtimes of quantum programs in finite-dimensional state spaces is developed. Several examples are provided as applications of this method, including computing the expected runtime of quantum Bernoulli Factory -- a quantum algorithm for generating random numbers. In particular, using our new method, an open problem of computing the expected runtime of quantum random walks introduced by Ambainis et al. (\textit{STOC} 2001) is solved.
\end{abstract}



\maketitle

\section{Introduction}

Over the last few years, one has seen exciting progress in quantum computing hardware, such as Google's 53-qubit Sycamore \cite{Google19}, IBM Q \cite{IBM} and USTC's 62-qubit Zu Chongzhi \cite{USTC22}. This progress further motivates one to expect that the Noisy Intermediate-Scale Quantum (NISQ) technology \cite{preskill2018quantum} will find some practical applications in the next 5-10 years.

\begin{sloppypar}
\textbf{Quantum Resource and Runtime Estimation}: Resource and runtime estimation will be particularly important in programming NISQ devices because they are too resource-constrained to support error correction, and running long sequences of operations on them is impractical. On the other hand, resource and runtime estimation may help us understand the separation between the quantum algorithms that can be run on NISQ devices and those that must wait for a larger quantum computer. Indeed, quantum resource and runtime estimation problem has already attracted the attention of quantum computing researchers; for example, 
resource and timing analysis was incorporated into quantum compilation framework ScaffCC \cite{Ali15}; 
an estimation of required qubits and quantum gates for Shor's algorithm to attack ECC (Elliptic Curve Cryptography) was given in \cite{roetteler2017quantum}. 
But current research in this area has been carried out mainly in a manner of case by case. Certainly, a more principled approach to this problem would be desirable.  
\end{sloppypar}

\textbf{Techniques in Probabilistic Programming}: Recently, a series of powerful techniques for resource and runtime estimation of probabilistic programs have been proposed (see for example \cite{celiku2005compositional,mciver2005,chakarov2013,brazdil2015,olmedo2016,chatterjee2017,batz2018}).
In particular, inspired by Nielson's Hoare-like proof system for reasoning about the running times of non-probabilistic programs \cite{nielson1987hoare}, a weakest precondition calculus was developed in \cite{kaminski2016,kaminski2018} for analysing the expected runtimes of randomised algorithms and probabilistic programs. It has been successfully applied to estimate the expected runtimes of several interesting example probabilistic programs, including the coupon collector's problem, one-dimensional random walk and randomised binary search.  
Furthermore, an analysis that can derive symbolic bound on the expected resource consumption of probabilistic programs was presented in \cite{Ngo2018} by effectively combining the weakest precondition reasoning for probabilistic programs 
\cite{vcerny2015segment,kaminski2016,dehesa2017verifying,vasconcelos2015type} with the automatic amortised resource analysis (AARA) for non-probabilistic programs \cite{carbonneaux2017automated,celiku2005compositional,hoffmann2015type,hofmann2015multivariate}. The strength of this approach is that it can be fully automated by reducing the bound analysis to LP (Linear Program) solving, and its effectiveness was demonstrated by automatic analysis of a large number of challenging randomised algorithms and probabilistic programs.

\textbf{From Quantum Weakest Preconditions to Quantum Expected Runtimes}: The well-known statistical nature of quantum systems immediately suggests the possibility of extending the techniques discussed above for solving the corresponding problems in quantum computing. Fortunately, a foundation for this research was already laid in the seminal work \cite{d2006quantum} where the notion of weakest quantum precondition was defined. The aim of this paper is to extend the line of research initiated by \cite{d2006quantum} and to develop a weakest precondition calculus for reasoning about the expected runtimes of quantum programs. Our basic ideas in the extension from quantum weakest precondition to quantum expected runtimes are largely inspired by  
the results achieved in the studies of expected runtimes of probabilistic programs \cite{kaminski2016,kaminski2018}. But several challenges exist in the transition from the probabilistic case to the quantum case:
\begin{itemize}\item \textit{Conceptual challenge}: The expected runtime $\mathit{\ert}(S)$ of a probabilistic program $S$ is defined in \cite{kaminski2016,kaminski2018} as a transformer of runtime functions, which are modelled as mappings from the state space to nonnegative real numbers or $\infty$ that are 
linear on the probabilistic combination of states. 
  Whenever generalised to the quantum case, 
  expected runtime functions are no longer linear on the superposition of pure states due to the Born rule in quantum mechanics. Thus, we need to find 
an appropriate interpretation for them as observables in quantum physics (or mathematically as Hermitian operators) so that they can be effectively manipulated and computed. 
  We resolve this issue using several mathematical techniques developed in the previous work on quantum weakest preconditions \cite{d2006quantum} and quantum Hoare logic \cite{Ying11}.  

\item \textit{Computational challenge}: Although quantum gates are modelled as unitary operators, quantum measurements are used in the guards of conditional statements and loops. Thus, super-operators are inevitably involved in their denotational semantics and computing their weakest preconditions. However, super-operators are (completely positive) mappings from (linear) operators to themselves and are much harder to manipulate than ordinary operators. Fortunately, some matrix representations of super-operators in Hilbert spaces with finite dimensions have been developed in quantum physics. We are able to generalise this technique (as tailored in \cite{Ying13}; also see \cite{Ying16}, Section 5.1.2)
into our weakest precondition calculus for expected runtimes of quantum programs.   
\end{itemize}

\textbf{Contribution of this Paper}: This paper presents a weakest  precondition approach for reasoning about the expected runtimes of  quantum  programs by overcoming the above two challenges. More concretely, we achieve the following contributions: 
\begin{enumerate}
  \item By generalising the notion of quantum weakest precondition \cite{d2006quantum}, we formally define the expected runtimes 
    of quantum \textbf{while}-programs. In particular, a quantum  
    observable representation of them is given (Theorem \ref{main_th1}). The significance of this representation is three-fold. Firstly, it provides a physical interpretation of the notion of expected runtime.
    \item Secondly, as a corollary of Theorem \ref{main_th1}, for quantum programs in finite-dimensional Hilbert spaces, we show that
    almost surely terminating on an input state is equivalent to positive 
    almost-sure termination (i.e. with a finite expected runtime) on the same state (Theorem \ref{main_th2}). 
    This result is less evident than its counterpart for probabilistic programs 
    because the interference of amplitudes in quantum computing makes it hard to 
    track the evolution of states of quantum programs.
  \item Thirdly, an effective method for computing the expected runtimes of quantum programs in finite-dimensional Hilbert spaces can be developed based on Theorem \ref{main_th1}. More explicitly, 
  by combining the observable representation with the matrix representation of super-operators, 
    we develop a method for computing the expected runtimes of quantum programs. This method works both numerically 
    and symbolically, with the numerical computation fully automated.
    It is worth mentioning that our method can deal with the case of infinite 
    execution paths, which was excluded in the method of \cite{Ali15}.
  \item The effectiveness of our method is tested by several case studies, including a key step of quantum Bernoulli factory for random number generation. 
  In particular, we solve an \textit{open problem} of computing the expected runtime of the quantum walk on
    an $n$-circle for any $n$ and an arbitrary initial state. The previously known result about this problem is that the expected runtime is $n$ when starting in a basis state for $n<30$ (see \cite{Ying13} or Section 5.1.3 of \cite{Ying16}). 
     \end{enumerate}
     
\textbf{Related Works}: The research on formal reasoning about expected runtimes of quantum programs was initiated in    \cite{olmedo2019runtime}. The aim of \cite{olmedo2019runtime} is the same as that of this paper, but the technical results of the two papers are very different. After giving a weakest precondition-style definition of expected runtime for quantum programs, \cite{olmedo2019runtime} focused on a case study of analyzing the runtime of (a simplified version of) the BB84 quantum key distribution protocol, and our  contributions 2) - 4) listed above were not considered there. Perhaps, a more essential difference between \cite{olmedo2019runtime} and our work is that     
in \cite{olmedo2019runtime}, quantum programs are treated as a kind of probabilistic programs and their expected runtimes are evaluated by  directly applying the techniques for probabilistic
programs. We argue that such a  generalization does not make use of the unique properties of quantum systems, and thus cannot be used to handle programs that have
quantum-specific behaviours.
To be more specific, let us consider the example quantum walk, where quantum interference happens. Using the definition of \cite{olmedo2019runtime}, a fixed point equation can be derived whose solutions are upper bounds of the expected runtime of the quantum walk, but solving this equation is hard because the quantum walk is treated as a probabilistic program and thus its quantum-specific feature is lost. In contrast, in our approach, a different fixed point equation is derived, and it can be solved by sufficiently exploiting mathematical properties of the involved quantum operations (as super-operators).


\textbf{Organization of the Paper}: 
We start from several working examples in Section \ref{sec_ex_2}, and then the syntax and semantics of quantum programs are reviewed in Section \ref{sec:quantum_while_program}. 
The expected
runtime of quantum programs are defined in Section \ref{sec:definition_of_runtime_obervable}. The observable 
representation is presented in Section \ref{sec-represent}. Based on it, we prove the equivalence of almost surely 
termination and positive almost-sure termination of quantum programs in Section \ref{sec:condition_of_finite_expected_runtime}.
A method for computing the expected runtimes of quantum programs is presented in Section \ref{sec:computing_expected_runtime_observable}. The case studies are given in Section \ref{sec:case_studies}.
For the sake of readability, proofs of our results can be found in the appendices.

\section{Quantum Programs}
\label{sec:quantum_programs}

In this section, we set our stage by reviewing the syntax and semantics of the quantum programs considered in this paper. We assume that the reader is familiar with basic ideas of quantum computing; otherwise, she/he can consult the standard textbook \cite{NC} or the preliminary sections of quantum programming literature, e.g. \cite{SZ00,Selinger04}.

\subsection{Working Examples}
\label{sec_ex_2}

Let us start with several simple examples. They are deliberately chosen as the quantum analogues of some examples considered in \cite{kaminski2016,kaminski2018} so that the reader can observe the similarity and subtle differences between probabilistic and quantum programs.

We first consider the process that keeps tossing a fair coin until the first head occurs.
  This process can be described as a probabilistic program:
  \begin{equation}
   P_\mathit{geo}\equiv \while\ (c = 1) \{c := 0 \oplus_{\frac{1}{2}} c := 1\}
  \end{equation}
  where $0, 1$ are used to indicate the heads and tails, respectively, and $P_1 \oplus_a P_2$ stands for a probabilistic choice that chooses to
  execute $P_1$ with probability $a$ and to execute $P_2$ with probability
  $1 - a$. A quantum analogue of this program is given as the following:
\begin{example}[Quantum Geometric Distribution]
  \label{ex-coin} A quantum coin can be modelled as a quantum bit (qubit for short) variable $q$. Using the Dirac notation, we write $|1\rangle, |0\rangle$ for the head and tail, respectively. Quantum coin $q$ can be in states $|1\rangle, |0\rangle$ as a classical coin, but it can also be in a superposition  $a|0\rangle+b|1\rangle$, where $a,b$ are complex numbers satisfying the normalisation condition $|a|^2+|b|^2=1$.  
  Thus, its state space is the $2$-dimensional Hilbert space
  $\HH_q=\mathit{span}\{|0\rangle, |1\rangle\}$ with head $|1\rangle$ and tail $|0\rangle$ as its basis states. To detect the state of the quantum coin, we need to perform quantum measurement on it. Here, we use the measurement $M$ in the computational basis $\{|0\rangle, |1\rangle\}$, which is mathematically modelled as a pair of operators  $M=\{M_0,M_1\}$ 
  with $M_0=|0\rangle\langle 0|$ and $M_1=|1\rangle\langle 1|$.
  For example, when performing $M$ on a quantum coin in superposition  $a|0\rangle+b|1\rangle$,  we will see the tail $|0\rangle$ with probability $|a|^2$ and the head $|1\rangle$ with probability $|b|^2$. 
  
  A quantum program that behaves in a way similar to $P_\mathit{geo}$ is
  defined as: 
\begin{equation}
  \label{q-coin}Q_\mathit{geo}\equiv\mathbf{while}\ M[q]=1\ \mathbf{do}\ q:=H[q]\ \mathbf{od}
\end{equation}
  where $H$ is the Hadamard gate represented by the matrix 
  $$H=\frac{1}{\sqrt{2}}\left(\begin{array}{cc}1 & 1\\ 1 & -1\end{array}\right)$$ 
  which transforms the tail $|0\rangle$ and head $|1\rangle$ into their equal superpositions $|+\rangle= \frac{1}{\sqrt{2}}(|0\rangle+ |1\rangle)$ and $|-\rangle= \frac{1}{\sqrt{2}}(|0\rangle- |1\rangle).$\end{example}
  
It is interesting to compare programs $P_\mathit{geo}$ and $Q_\mathit{geo}$ carefully. First, coin tossing is treated in $P_\mathit{geo}$ as an abstract program construct $\oplus_{\frac{1}{2}}$ without specifying how to implement it. However, in $Q_\mathit{geo}$ we have to explicitly describe it as a physical process: it is realised by first applying the Hadamard gate $H$ on quantum coin $q$ and then measuring $q$ in the computational basis. Second, the coin $c$ in $P_\mathit{geo}$ is always in either state $0$ (tail) or $1$ (head) although we cannot predict its next state with certainty before tossing it. But if the quantum coin $q$ is initialised in tail $|0\rangle$ (respectively, head $|1\rangle$), then the Hadamard gate will transform it into a superposition $\ket{+}$ (or $\ket{-}$) of the head and tail. Third, checking the loop guard in $P_\mathit{geo}$ does not change the state of coin $c$. However, in $Q_\mathit{geo}$ we need to perform measurement $M$ on quantum coin $q$ when checking the loop guard in order to acquire information about $q$ and the measurement will change the state of $q$; for example, if $q$ is in state $|+\rangle$ before the measurement, then after the measurement $q$ will be in state $|0\rangle$ with probability $\frac{1}{2}$ and in $|1\rangle$ with probability $\frac{1}{2}$, but no longer in $|+ \rangle$.
It is even more interesting to note that there are (uncountably) infinitely many operators and measurements rather than $H$ and $M$ in $Q_\mathit{geo}$ suited to implement the \textit{fair} coin tossing.

\begin{sloppypar}
Another example considered in \cite{kaminski2016} is the following one-dimensional random walk with an absorbing boundary:
\begin{equation}\label{c-rw}C_\mathit{rw}\equiv\ x:=10;\ \while\ (x>0)\{x:=x-1\oplus_{\frac{1}{2}} x:=x+1\}
\end{equation} It starts from position $10$ and shifts on a line to the left or right with equal probability in each step until reaching the boundary at position $0$.
A quantum counterpart of this random walk, namely the semi-infinite
Hadamard walk, was defined in \cite{Ambainis2001}.
\end{sloppypar}

\begin{example}[Hadamard Walk] \label{ex-rw} Let $q$ be a quantum
  coin with a $2$-dimensional state space $\HH_q$ spanned by orthogonal basis
  states $\ket{L}$ and $\ket{R}$, indicating directions Left and Right,
  respectively. We introduce a quantum variable $p$ with state space $\HH_\infty$
  spanned by orthogonal basis $\{ \ket{n}: n \in \ZZ \}$, where $\ZZ$ is the set of integers, and $|n\rangle$ is used to denote position $n$ on a line. Then the Hadamard walk is considered as a composite system of $q$ and $p$ and thus its state space is the tensor product  $\HH\stackrel{\triangle}{=}\HH_q \otimes \HH_\infty$. A measurement that determines
  whether the system is at position $0$ is described as $N = \{ N_0, N_1\}$
  where $N_0 = \ket{0}_p\bra{0}$ and $N_1 = I_p - N_0$, and $I_p$ is the identity operators on $\HH_\infty$. To describe one step of the walk, we define a shift operator $S$  by
  \begin{equation*}
    S\ket{L, n} = \ket{L, n-1}, \quad S\ket{R, n} = \ket{R, n+1},
  \end{equation*}
  (together with linearity), meaning that the position shifts one position to the left or to the right  according to the state $|L\rangle$ or $|R\rangle$ of coin $q$.
  Intuitively, the quantum walk repeatedly behaves as follows:
  \begin{enumerate}
    \item \label{qw_loop}Perform measurement $N$ to see whether the system is at position $0$.
    \item If it is at $0$, then terminate; otherwise, apply Hadamard gate $H$ to quantum coin $q$ in order to generate an equal superposition of the directions Left and Right, and then move the position according to shift
      operation $S$ and goto step \ref{qw_loop}.
  \end{enumerate}
Formally, it can be written as a quantum program:
\begin{equation}
  \label{q-rw}
Q_\mathit{rw}\equiv \mathbf{while}\ N[p]=1\ \mathbf{do}\ q:=H[q];\ q,p:=S[q,p]\ \mathbf{od}.
\end{equation}
\end{example}

An apparent difference between $C_\mathit{rw}$ and $Q_\mathit{rw}$ is that the latter can move to the left and right simultaneously; for example, if currently the coin is in state $|R\rangle$ and the position is $n$, then applying Hadamard gate $H$ yields a superposition $\frac{1}{\sqrt{2}}(|L\rangle-|R\rangle)$ of directions $L$ and $R$, and shift operator $S$ transforms to system to state $\frac{1}{\sqrt{2}}(|L,n-1\rangle-|R,n+1\rangle).$ A more essential difference between $C_\mathit{rw}$ and $Q_\mathit{rw}$ is the so-called \textit{quantum interference}: the coefficients (called probability amplitudes) in a quantum state can be negative (and imaginary) numbers (see the above example states), and thus two paths of the walk $Q_\mathit{rw}$ with positive and negative amplitudes respectively can cancel one another. This feature has been extensively exploited to design quantum walks-based algorithms faster than the corresponding classical algorithms (see for example, \cite{Dorit2001,Kempe2003,Andris2003,Andris2008}).

The following probabilistic program was used in \cite{kaminski2016} to show a fundamental difference
between the runtime of non-probabilistic programs and that of probabilistic programs:
 \begin{equation}\begin{split}
C\equiv  C_1: &\ x := 1;\ c := 1;\\
&\while\ (c = 1) \{ c := 0 \oplus_{\frac{1}{2}} c := 1; \  x := 2x \} \\
  C_2: &\ \while\ (x > 0) \{ x := x - 1 \}
\end{split}\end{equation}
Obviously, both of subprograms $C_1$ and $C_2$ have a finite expected runtime on all
inputs. However, it is easy to see that $C_1;C_2$ has an
infinite expected runtime. A quantum variant of this program is given as the following:
\begin{example}
  \label{ex_2_3}
  Let quantum variables $q$, $p$, Hilbert space $\HH_q$, $\HH_\infty$ and measurement $N$ be the same as in Example \ref{ex-rw}. We introduce the following two unitary operators on $\HH_\infty$ as quantum mimics of assignments $x:=x-1$ and $x:=2x$, respectively, in the above program $C$:\begin{itemize}\item The left-shift operator $T_L$ is defined by $T_L|n\rangle=|n-1\rangle$ for every $n\in\ZZ$.\item The duplication operator $D$ is defined as follows: $D|n\rangle=|2n\rangle$ for $n\geq 0$ and $D|n\rangle=|K(n)\rangle$ for $n<0$. To make $D$ be a unitary operator on $\HH_\infty$, $K$ must be chosen as a one-onto-one mapping: $$\{-n:n\geq 1\}\rightarrow\{-n:n\geq 1\}\cup\{2n-1:n\geq 1\},$$ e.g. $K(1-2n)=2n-1$ and $K(-2n)=-n$ for $n\geq 1$.\end{itemize} Then we can define quantum program:
\begin{equation}\label{q-two}\begin{split}Q\equiv\ &Q_1: \mathbf{while}\ M[q]=1\ \mathbf{do}\ q:=H[q];\ p:=D[p]\ \mathbf{od};\\
&Q_2: \mathbf{while}\ N[p]=1\ \mathbf{do}\ p:=T_L[p]\ \mathbf{od}
\end{split}
\end{equation}
\end{example}

The above quantum program $Q$ behaves similarly to probabilistic program $C$ when the input state is $\ket{R, 1}$.
It is worth noting an interesting difference between the left-shift operator $T_L$ in $Q$ and the shift operator $S$ in Example \ref{ex-rw}. Since $T_L$ always moves in the same left direction, it can be defined simply on $\mathcal{H}_\infty$. In contrast, $S$ moves in the direction determined by the quantum coin $q$, so it has to be defined on $\mathcal{H}_d\otimes\mathcal{H}_\infty.$  

The three examples presented above may give the reader the impression that quantum programs are very similar to probabilistic programs. But we must point out that quantum programs are notoriously harder to analyze than probabilistic ones, as we will see from our example of the quantum walk in Section  \ref{sub:quantum_random_walk}.


\subsection{Syntax}\label{sec:quantum_while_program}

The examples presented in the above subsection should give the reader intuition about basic quantum program constructs. In this subsection, we formally define the syntax of quantum programs studied in this paper.

We choose to use the quantum \textbf{while}-language defined in \cite{Ying11,Ying16} for a consistency with \cite{kaminski2016,kaminski2018}, where probabilistic \textbf{while}-language was employed. We assume a countably infinite set $\qvar$ of quantum variables and use $q, q_0, q_1, q_2, \dots$ to denote them.
The state Hilbert space of a quantum variable $q$ is denoted
$\HH_q$. A quantum register is a finite sequence of distinct quantum variables.
The state space of a quantum register $\oq = q_0 \dots q_n$ is then the tensor product $ \HH_{\oq} = \bigotimes_{i = 0}^{n} \HH_{q_i}.$

\begin{definition}[Syntax \cite{Ying11}] The set $\qprogs$ of quantum \textbf{while}-programs is defined by the following syntax:
\begin{align}
  S :: = \ & \skp \mid  S_1; S_2 \mid q := \ket{0} \mid \oq := U [\oq]  \label{progc2}\\
              &\mid \ifb\ (\guard m\cdot M[\oq] = m \to S_m)\ \ife \label{progc3}\\
              &\mid \while\ M[\oq]=1\ \wdo\ S\ \wod \label{progc4}
\end{align}\end{definition}

A brief explanation of the above program constructs is given as follows.
The constructs $\skp$ and sequential composition $S_1; S_2$ are similar to their counterparts in the classical or probabilistic \textbf{while}-programs. 
The initialisation $q := \ket{0}$ sets the quantum register $q$ to the basis state $\ket{0}$. 
The statement $\oq := U [\oq]$ means that unitary transformation $U$ is performed on the quantum register $\oq$. 
The construct in (\ref{progc3}) is a quantum generalisation of the classical case 
statement. In the execution, measurement $M = \{ M_m \}$ is performed on $\oq$, 
and then a subprogram $S_m$ will be selected according to the measurement outcome. 
The statement in (\ref{progc4}) is a quantum generalisation of \textbf{while}-loop, where 
the measurement $M$ has only two possible outcomes: if the outcome is
$0$, the program terminates, and if the outcome $1$ occurs, the program
executes the loop body $S$ and then continues the loop. Most of these constructs were already used in our working Examples \ref{ex-coin}, \ref{ex-rw} and \ref{ex_2_3}.  

\subsection{Semantics}
For each quantum program $S$, we write $\mathit{var}(S)$ for the set of all variables $q\in\qvar$ appearing in $S$.
The Hilbert space of program $S$ is the tensor product $\HH_S = \bigotimes_{q \in \mathit{var}(S)} \HH_{q}.$
Let $\DD(\HH_S)$ be the set of all partial density operators (i.e. positive operators with the trace $\leq 1$) on $\HH_S$. 
A state of  program $S$ is then represented by a partial density operator $\rho \in \DD(\HH_S)$. Furthermore, a configuration is defined as a pair $(S,\rho)$ of a program $S$ and a state $\rho$. The operational semantics of quantum programs can be defined as a transition relation between configurations. 
Based on it,     
 the following denotational semantics can be  derived, and will be extensively used in this paper:
\begin{lemma}[Structural Representation of Denotational Semantics  \cite{Ying11}]\label{lem-structural} For any input state $\rho \in \HH_S$, we have: 
\begin{enumerate}
  \item \label{dsem_skp} $\sem{\skp}(\rho) = \rho$;
  \item \label{dsem_init} $\sem{q := \ket{0}}(\rho) = \sum_n|0\rangle_q\langle n|\rho|n\rangle_q\langle 0|$;
  \item \label{dsem_uni} $\sem{\oq := U\oq}(\rho) = U \rho U^\dagger$;
  \item \label{dsem_comp} $\sem{S_1; S_2}(\rho) = \sem{S_2}(\sem{S_1}(\rho))$; 
  \item \label{dsem_if} $\sem{\ifb(\guard m\cdot M[\oq] = m \to S_m)\ife}(\rho) = \sum_m \sem{S_m}(M_m \rho M_m^\dagger)$; 
    \item \label{dsem_while} for loop $\while[M, S]\equiv \while\ M[\oq]=1\ \wdo\ S\ \wod$: 
    $$\sem{\while[M, S]}(\rho) = \bigsqcup_{k = 0}^{\infty} \sem{\while^{(k)}[M,S]}(\rho),$$ 
where $\while^{(k)}[M,S]$ is the $k$-fold iteration of the loop $\while$: 
\begin{equation*}\label{iteration}\begin{cases}
  &\while^{(0)}[M,S]  \equiv \textbf{abort}, \\
  &\while^{(k+1)}[M,S] \\
  &\begin{aligned} 
  \equiv\ &\ifb\ M[\oq]  =\  0 \to \skp\\  
  &\guard\qquad \qquad 1\to S; \while^{(k)}[M,S]\ \ife
\end{aligned}\end{cases}\end{equation*}
for $k\geq 0$, $\bigsqcup$ stands for the least upper bound in the CPO of partial density operators with the L\"owner order $\sqsubseteq$ (see \cite{Ying16}, Lemma 3.3.2), 
and $\textbf{abort}$ is a program that never terminates so that $\sem{\textbf{abort}}(\rho) = \bf{0}$ for all $\rho$.
\end{enumerate}\end{lemma}

We remark that a partial density operator $\rho \in \HH_S$ can be thought of as a representation of a sub-distribution of pure states in $\HH_S$. Thus,  the sum of partial density operators in the above lemma is analogous to the sum of sub-distributions of states in probabilistic programming.

It immediately follows from the above lemma that the denotational semantics $\sem{S}$ of a quantum program is a (completely positive) super-operator, which is the mathematical formalism of quantum operations or the (discrete-time) dynamics of open quantum systems. 

\section{Expected Runtime as a Real-Valued Function}
\label{sec:definition_of_runtime_obervable}

In this section, we formally define the expected runtime function of a quantum program $S$ as a real-valued function that maps an initial state to the expected runtime on that state. This notion will be further illustrated  using our working examples given in Section \ref{sec_ex_2}. 

The definition of expected runtime given in the present section is intuitive because it is obtained by generalising the ideas of \cite{kaminski2016,kaminski2018} from the probabilistic case to the quantum case. In the next section, an elegant representation of the expected runtime function will be developed in terms of observables, and thus a link to quantum weakest precondition \cite{d2006quantum} is established. 


\subsection{Definition}

We only consider the case where the state Hilbert spaces of all quantum variables in a  quantum program $S$ are finite-dimensional, and thus $\HH_S$ is finite-dimensional too. 

To simplify the presentation, let us first introduce several notations. For the measurement $M=\{M_m\}$ in \textbf{if} statement (\ref{progc3}), and for each possible measurement outcome $m$, we define a super-operator $\EE_{M_m}$ by $\EE_{M_m}(\rho)=M_m\rho M_m^\dagger$ for all density operators $\rho$. Similarly, for the measurement $M=\{M_0,M_1\}$ in \textbf{while}-loop (\ref{progc4}), we define super-operators $\EE_0$, $\EE_1$ as $$\EE_0(\rho)=M_0\rho M_0^\dagger, \quad \EE_1(\rho)=M_1\rho M_1^\dagger$$ for all density operators $\rho$.  

A straightforward generalisation of the expected runtimes of probabilistic programs defined in \cite{kaminski2016,kaminski2018} yields the following: 

\begin{definition}\label{def_EXT}
  The expected runtime of a quantum program $S$ is a real-valued function $\ERT[S]: \DD(\HH_S) \to \RR \cup \{ \infty \}$ defined as follows: 
  \begin{enumerate}
    \item $\ERT[\skp](\rho) = 0$; 
    \item $\ERT[q := \ket{0}](\rho) = \tr{\rho}$; 
    \item $\ERT[\oq := U[\oq]](\rho) = \tr{\rho}$; 
    \item $\ERT[S_1; S_2](\rho) = \ERT(S_1)(\rho) + \ERT(S_2)(\sem{S_1}(\rho))$; 
    \item $\ERT[\ifb(\guard m\cdot M[\oq] = m \to S_m)\ife](\rho) = \\
    \phantom{ }\qquad\qquad\qquad\qquad\qquad\quad \tr \rho + \sum_{m} \ERT[S_m](\EE_{M_m}(\rho))$; 
    \item $\ERT[\while\ M[\oq]=1\ \wdo\ S\ \wod](\rho) = \\
    \phantom{ }\qquad\qquad\qquad\qquad\qquad\ \  \lim\limits_{k \to \infty} \ERT[\while^{[k]}[M, S]](\rho), $ \\ 
        where $\while^{[k]}[M,S]$ is the first $k$ iterations of the loop defined by:
\begin{equation*}
  \begin{cases}
  &\while^{[0]}[M,S]  \equiv \skp, \\
  &\while^{[k+1]}[M,S] \\
  &\begin{aligned}
  \equiv\ &\ifb\ M[\oq] = \ 0 \to \skp \\
  &\guard \qquad\qquad 1 \to S; \while^{[k]}[M,S]\ \ife
\end{aligned}\end{cases}\end{equation*}
for $k \geq 0$. [Note the difference between $\while^{[0]}$ defined here and $\while^{(0)}$ in Lemma \ref{lem-structural}. This makes $\while^{[k]}$ and $\while^{(k)}$ different for all $k\geq 0$.]
  \end{enumerate}
\end{definition}

Intuitively, for any state $\rho\in \DD(\HH_S)$, $\ERT[S](\rho)$ is the expected runtime of program $S$ on input $\rho$. Our design decisions in the above definition are explained as follows:    
\begin{itemize}
  \item As usual, we choose to assume the expected runtime of $\skp$ to be zero. This will be used in defining the expected runtime of a $\while$-loop. It should be noted that the runtime of $\skp$ is not equal to that of  
    identity transformation $q := I[q]$, which does nothing but takes $1$ step.
  \item When applying to a density operator, the runtime of an initialisation or a unitary transformation is defined to be $1$. 
    We would like to extend the domain of $\ERT[S]$ from density operators to
    partial density operators for simplifying the presentation. It is reasonable to define their runtime applying to a partial density operator $\rho$ as $\tr(\rho)$
    , which is the product of probability $\tr \rho$ and $1$.
    Such a definition is consistent with the requirement  that $\ERT[S]$ is linear.  
  \item The runtime of sequential composition $S_1;S_2$ is defined in the same way as the case of probabilistic programs. It is worth noting that whenever $\ERT[S_1](\rho) = \infty$, then the second summand $\ERT[S_2](\sem{S_1}(\rho))$ can be ignored.
  \item The runtime of an $\ifb$ statement is defined as the sum of the 
    runtimes of all branches plus $\tr(\rho)$, which can be thought of as
      the time that the measurement in its guard takes.
  \item The runtime of a $\while$ statement is defined as the
    limit of the runtime of its unfolding (iterations).  
\end{itemize}

 The function $\ERT$ for probabilistic programs was derived in \cite{kaminski2016, kaminski2018} from their operational semantics. Here we choose to define $\ERT$ for quantum programs directly. Indeed, in the quantum case, $\ERT$ can also be derived from the operational semantics (see Section 3.2 of \cite{Ying16}), but the derivation is quite tedious. We believe that a direct definition can make our main idea clearer, and leave the derivation in the Appendix \ref{sub:derivation_of_ert_from_operational_semantics} for interested readers.

We now present several useful properties of the function $\ERT$. The following lemma gives a way for computing the runtime of the iterations of a loop. 
\begin{lemma}
  \label{lemm_ERT_while} For any input state $\rho$, we have: 
  \begin{align*}
    &\ERT[\while^{[n]}[M, S]](\rho) =  \sum_{k = 0}^{n-1} \tr \left( (\sem{S}\circ \EE_1)^{k}(\rho) \right)\\ 
    & \qquad\qquad\qquad\qquad + \sum_{k = 0}^{n - 1} \ERT[S](\EE_1 \circ (\sem{S} \circ \EE_1)^{k}(\rho)). 
  \end{align*}
  where $\circ$ stands for the composition of super-operators.
\end{lemma}

Essentially, the $k$-th term of the first part in the right-hand side of the above equation is the time that the $k$-th measurement $M$ takes in the process of iterations, and the $k$-th term of the second part is the time that the $k$-th application of $S$ takes. 

We define the set of inputs from which the expected runtime of program $S$ is finite: \begin{equation}\label{finite-run}T_S = \{ \rho \in \DD(\HH_S): \ERT[S](\rho) < \infty \}.\end{equation} The following lemma
shows the linearity of the expected runtime of quantum programs over $T_S$.

\begin{lemma}[Linearity]\label{linearity}  For any $\rho_1, \rho_2 \in T_S$ and $\lambda_1, \lambda_2 > 0$ with $\lambda_1 + \lambda_2 \leq 1$: 
  $$\ERT[S](\lambda_1 \rho_1 + \lambda_2 \rho_2) = \lambda_1 \ERT[S](\rho_1) + \lambda_2 \ERT[S](\rho_2).$$
\end{lemma}

\subsection{Examples}
\label{sec_ex_ert}

To illustrate the definition introduced in the above subsection, let us compute the expected runtimes of the programs in our working examples presented in Section \ref{sec_ex_2}. Some more practical examples will be given in Section \ref{sec:case_studies} as case studies. 

\begin{example}
  \label{ex_geo}
  Consider repeated quantum coin tossing program $Q_\mathit{geo}$ defined in Eq. (\ref{q-coin}) with initial state $|1\rangle$. 
Denote command $q := H[q]$ by $S_H$. Then by definition we obtain: 
\begin{align*}
  \ERT[\while&^{[k+1]}[M, S_H]](\ket{1} \bra{1}) \\
  = &2 + \ERT[\while^{[k]}[M, S_H]](\ket{-} \bra{-})
\end{align*} where $\ket{-}= \frac{1}{\sqrt{2}}(|0\rangle-|1\rangle).$ 
Moreover, we have:
\begin{align*}
  \ERT[\while&^{[k+1]}[M, S_H]](\ket{-} \bra{-}) \\
  = &\frac{3}{2} + \ERT[\while^{[k]}[M, S_H]](\frac{1}{2} \ket{-} \bra{-}).
\end{align*}
Using Lemma \ref{linearity} we further obtain:
\begin{equation*}
  \ERT[\while^{[k+1]}[M, S_H]](\ket{1} \bra{1}) = 2 + \sum_{j = 1}^{k} \frac{3}{2^j}.
\end{equation*}
Therefore, we have:
\begin{equation*}
  \ERT[Q_\mathit{geo}] = \lim\limits_{k \to \infty} \ERT[\while^{[k+1]}[M, S_H]](\ket{1} \bra{1}) = 5.
\end{equation*}
This shows that the expected runtime of quantum program $Q_\mathit{geo}$ is the same as that of its probabilistic counterpart $C_\mathit{geo}$.
\end{example}

\begin{example}\label{ex-rw-EXT}
Consider semi-infinite Hadamard walk $Q_\mathit{rw}$ defined in Eq. (\ref{q-rw}) initialised in direction $c:=|L\rangle$ and position $q:=|1\rangle$. 
It was proved in \cite{Ambainis2001} that its termination probability is  
$2/\pi$. We can use the tools developed here to prove that its expected runtime is $\infty$. 
Let $\rho = \ket{L, 1}\bra{L, 1}$ be the density operator corresponding to the initial pure state. According to Theorem 8 in \cite{Ambainis2001}, we have:  
\begin{equation*}
  \sum_{k = 0}^{\infty} \tr(\EE_0 \circ (\sem{S}\circ \EE_1)^k (\rho)) = \frac{2}{\pi}. 
\end{equation*}
Then using Lemma \ref{lemm_ERT_while}, we derive:
  \begin{align*}
    &\ERT[Q_\mathit{qw}](\rho) 
     \geq \sum_{k = 0}^{\infty} \tr \left( (\sem{S}\circ \EE_1)^{k}(\rho) \right) \\
    & = \tr\rho + \sum_{k = 1}^{\infty} (\tr \rho - \sum_{j = 0}^{k-1}\tr(\EE_0 \circ (\sem{S}\circ \EE_1)^j (\rho))) \\
    & \geq 1 + \sum_{k = 1}^{\infty}(1 - \frac{2}{\pi}) = \infty.
  \end{align*}
\end{example}

\begin{example}
\label{ex_2_3_EXT}
Consider quantum program $Q\equiv Q_1;Q_2$ defined in Eq. (\ref{q-two}) with input $q:=|R\rangle$ and $p:=|1\rangle$. 
We show that similar to its probabilistic counterpart $C\equiv C_1;C_2$, it has 
an infinite expected runtime. 
Let $\rho_0 = \ket{R}_q \bra{R} \otimes \ket{1}_p \bra{1}.$ It can be shown in a way similar to Example \ref{ex_geo} that $$\ERT[Q_1](\rho_0) = 7.$$
Moreover, it can be proved by induction that
  $$ \sem{Q_1}(\rho_0) = \sum_{k=1}^{\infty} \frac{1}{2^k} \ket{L}_q \bra{L} \otimes \ket{2^k}_p \bra{2^k}.$$
Then we have  
  $ \ERT[Q_2](\ket{L}_q \bra{L} \otimes \ket{t}_p \bra{t}) = 2t+1 $
for $t > 0$.
Finally, we obtain:
\begin{equation*}
  \ERT[Q_1; Q_2](\rho_0) = 7 + \sum_{k = 1}^{\infty} \frac{1}{2^k} (2^{k+1} + 1) = \infty. 
\end{equation*}
\end{example}

\section{Expected Runtime as an Observable}
\label{sec:expected_runtime_as_an_observable}

In the previous section, the expected runtime function $\ERT[S]$ of a quantum program $S$ is simply seen as a real-valued function from input states. This is consistent with the idea of expected runtimes as generalised weakest preconditions in \cite{kaminski2016,kaminski2018} for probabilistic programs. However, quantum weakest preconditions are defined in  \cite{d2006quantum} as physical observables. The aim of this section is to present a representation of expected run time $\ERT(S)$ as a physical observable, through which the close connection between $\ERT(S)$ and quantum weakest preconditions becomes clear. Such a   representation brings an additional benefit. Obviously, it is inefficient to compute $\ERT[S](\rho)$ by Definition \ref{def_EXT} for each input $\rho$, especially when $S$ contains loops. 
The representation of $\ERT[S]$, however, enables us to find a uniform way of computing expected runtime $\ERT[S](\rho)$ for all inputs $\rho$, which will be presented in the next section.  

\subsection{Representation Theorem}\label{sec-represent}
The notion of observable is a cornerstone of quantum mechanics. Recall from \cite{NC} (or any standard textbook of quantum mechanics) that an observable of a quantum system with state Hilbert space $\HH$ is mathematically described by a Hermitian operator $A$ on $\HH$. It determines a quantum measurement with the  eigenvalues of $A$ being the possible outcomes. If the system's current state is $\rho$ and we perform the measurement on it, then nondeterminism may happen here: different outcomes may occur with certain probabilities. But we have a statistical law; that is, one can assert that the expectation (average value) of the outcomes is $\mathit{tr}(A\rho)$.

{\color{red}
}

Now let us consider how $\ERT(S)$ can be characterised in terms of physical observable. First of all, we notice that for each quantum program $S$, the linearity of $\ERT(S)$ shown in Lemma \ref{linearity} immediately implies that there exists 
an observable $A_S$ such that \begin{equation}\label{gap1}\ERT(S)(\rho) = \tr(A_S \rho)\end{equation} for all 
$\rho \in T_S$ (the set of input states with finite expected runtimes; see Eq. (\ref{finite-run})). So, our actual problem is to find an explicit and structural representation of this observable $A_S$. 

To present the representation of $A_S$, we need the following lemma proved in \cite{Zhou19}. Recall from \cite{Ying11} that $\tr(\sem{S}(\rho))$ is the probability that quantum program $S$ terminates with input $\rho$. Then for a density operator $\rho$, we say that $S$ almost surely terminates upon $\rho$ if $\tr(\sem{S}(\rho))=1$. By linearity of trace and $\sem{S}$, this condition should be rewritten as $\tr(\sem{S}(\rho))=\tr(\rho)$ if $\rho$ is a partial density operator.  
\begin{lemma}[Theorem 3.1 in \cite{Zhou19}]
  \label{lemm_term_space}  For any quantum \textbf{while}-program with state Hilbert space $\HH$, the set of the (unnormalised) pure states from which $S$ 
almost surely terminates:
  \begin{equation}\label{almost-term}
    V_S = \{ \ket{\psi} \in \HH: \tr(\sem{S}(\ket{\psi}\bra{\psi})) = \tr (\ket{\psi} \bra{\psi}) \}
  \end{equation}
  is a subspace of $\HH$
,which can be computed using the Kraus operator-sum representation of $\sem{S}$.
\end{lemma}

We will also need the notion of the dual of a super-operator. 
For any (completely positive) super-operator $\EE$ on Hilbert space $\HH$ such that $$\EE(\rho) = \sum_{j} E_j \rho E_j^\dagger,$$ its dual $\EE^*$ is an operator on Hermitian operators defined by $$\EE^*(A) = \sum_j E_j^\dagger A E_j.$$ The duality between operators $\EE$ and $\EE^\dagger$ is characterised by the equation $\tr(A \EE(\rho)) = \tr(\EE^*(A) \rho)$.

Now we are ready to define the promised representation of observable $A_S$. 



  \begin{definition} 
   \label{def_ert}
    Let $A$ be an observable. Then the expected runtime of quantum program $S$ followed by a continuation with expected runtime represented by $A$ is inductively defined as follow:
\begin{enumerate}
  \item $\ert[\skp](A) = A$.
  \item $\ert[q := \ket{0}](A) = I + \sum_{n} \ket{n}_q\bra{0}A\ket{0}_q\bra{n}$.
  \item $\ert[\oq := U[\oq]](A) = I + U^\dagger A U$.
  \item $\ert[S_1; S_2](A) = \ert[S_1](\ert[S_2](A))$.
  \item $\ert[\ifb(\guard m\cdot M[\oq] = m \to S_m)\ife](A) = \\
  \phantom{ }\qquad\qquad\qquad\qquad\qquad\qquad\quad I + \sum_m \EE_{M_m}^*(\ert[S_m](A))$.
  \item \label{def_ert_while} for loop $\while\equiv \while\ M[\oq]=1\ \wdo\ S\ \wod$: 
      \begin{align*}
      & \ert[\while](A) = \sum_{k = 0}^{\infty} P_\while (\EE_1^* \circ \sem{S}^*)^{k}(I) P_\while + \EE_0^*(A)\\ 
                       &\quad + \sum_{k = 0}^{\infty} P_\while (\EE_1^* \circ \sem{S}^*)^{k} \circ \EE_1^* (\ert[S](\EE_0^*(A))) P_\while\end{align*}
  \end{enumerate}
  where $P_\while$ is the projection onto the almost surely terminating subspace $V_\while$ of loop $\while\equiv \while\ M[\oq]=1\ \wdo\ S\ \wod$, as defined in Lemma \ref{lemm_term_space}.
\end{definition}

  We note that $\ert[\while](A)$ always converges, as guaranteed by the following lemma. 

  \begin{lemma}
    \label{ert_while_converge}
    Let $\while \equiv \while\ M[\oq]=1\ \wdo\ S\ \wod$ be a quantum loop and $X$ be a positive operator. Then
    $$ \sum_{k=0}^\infty P_{\while} (\EE_1^* \circ \sem{S}^* )^k(X) P_{\while} $$
    converges.
  \end{lemma}

Similar to classical and probabilistic weakest-precondition style calculus,  Definition \ref{def_ert} is compositional for sequential programs. 
The main difference between $\ert[S](A)$ and its probabilistic counterpart given in Table 1 of \cite{kaminski2018} is that the former is a Hermitian operator, and the latter is a real-valued function. In particular, the unit of time in their definitions is represented by the identity operator $I$ and the constant $\mathbf{1}$, respectively.
On the other hand, from Section 4 of \cite{d2006quantum} and Section 4.2.2 of \cite{Ying16}, it is clear that whenever the unit of time vanishes, $\ert[S](A)$ degenerates to the weakest precondition $wp.S.A$ of $S$ given postcondition $A$.

By induction on the structure of $S$, it is easy to show that $\ert[S](A)$ defined above is a positive (and thus Hermitian) operator. Therefore, it is indeed (the mathematical description of) a physical observable. 

In particular, the observable $\ert[S](\zeromat)$ describes the expected runtime of program $S$ without any continuation. To simplify the presentation we abbreviate it to $\ert[S]$ in the absence of ambiguity. 

The following lemma clarifies the relationship between $\ert[S](A)$ and $\ert[S]$.

  \begin{lemma} \label{split_ert}
  For any observable $A$ and initial state $\rho \in W_S$, we have:
    \begin{equation*}
      \tr(\ert[S](A) \cdot \rho) = \tr(\ert[S] \cdot \rho) + \tr(A \cdot \sem{S}(\rho)).
    \end{equation*}
  \end{lemma}


Now we can present the main result of this section. Note that all elements of the almost surely terminating subspace $V_S$ defined by Eq. (\ref{almost-term}) are pure states. We further define the set of partial density operators upon which program $S$ almost surely terminates: 
 $$W_S = \{ \rho \in \DD(\HH) : \tr (\sem{S}(\rho)) = \tr \rho \}.$$ Then our representation of $A_S$ can be stated as the following: 
 
\begin{theorem}[Observable Representation of Runtime]
  \label{main_th1}
  For any quantum program $S$ and for all $\rho \in W_S$, we have: 
  \begin{equation} \label{gap2} 
    \ERT[S](\rho) = \tr (\ert[S] \cdot \rho).
  \end{equation}
\end{theorem}

The above theorem shows that observable $\ert[S]$ is an explicit representation of $A_S$ defined in Eq. (\ref{gap1}) over $W_S$.

\subsection{A Condition for Finite Expected Runtimes}
\label{sec:condition_of_finite_expected_runtime}

The reader must have already noticed a gap between Eqs. (\ref{gap1}) and (\ref{gap2}): the former is valid over $T_S$ and the latter is valid over $W_S$. 
This gap can actually be filled in by the following:

\begin{theorem}[Equivalence of Almost Sure Termination and Finite Expected Runtime]\label{main_th2}
  Let $S$ be an arbitrary quantum \textbf{while}-program with state Hilbert space $\HH$. Then for any partial density operator $\rho\in\mathcal{D}(\HH)$: 
  \begin{equation*}
    \ERT[S](\rho) < \infty \iff \tr(\sem{S}(\rho)) = \tr \rho.
  \end{equation*}
\end{theorem}

An immediate corollary of the above theorem is $T_S=W_S$. Moreover, it shows that  
a quantum program $S$ with input $\rho$ almost surely terminates if and only if it has a finite expected runtime.

\begin{sloppypar}
The above theorem further implies that if $\rho \in W_S$, then $\ERT[S](\rho) = \tr(\ert[S]\cdot \rho)$; otherwise $\ERT[S](\rho) = \infty$. However, calculating the value of $\ert[S]$ is challenging when $S$ has loops. This problem will be considered in Section \ref{sub:computing_expected_runtime_observable_of_while_statement}.
\end{sloppypar}

Remember that in this paper, we only consider quantum programs in finite-dimensional Hilbert spaces. Then we can 
furthermore give a finite upper bound for the expected runtime for a quantum
program $S$ for all almost surely terminating inputs $\rho \in W_S$. Recall that for an operator $A$ on a Hilbert space $\HH$, its norm is defined as
$$\|A\|=\sup\{|A|\psi\rangle|: |\psi\rangle\in\HH\ {\rm and}\ ||\psi\rangle|=1\}.$$

\begin{corollary}[Uniform Bound of Expected Runtime]
  For any quantum program $S$, let $M$ be the norm of its expected runtime observable: $M = \norm{\ert[S]}$.  Then for all $\rho \in W_S$, we have $$\ERT[S](\rho) \leq M.$$
\end{corollary}


\section{Computing Expected Runtime}
\label{sec:computing_expected_runtime_observable}

Theorem \ref{main_th1} established in the last section gives a physical interpretation of expected runtime as an observable. As pointed out before, this observable representation provides us with an efficient method for computing expected runtimes. More precisely, we see from Eq. (\ref{gap2}) that if the observable $\ert[S]$ is known, 
 we can easily compute the expected runtime $\ERT[S](\rho)$ for any  input $\rho$.
Therefore, in this section, we develop a method for computing $\ert[S]$ based on the matrix representation of denotational semantics of quantum programs, by 
  extending and incorporating the approach of \cite{Ying13} for termination analysis into the setting of expected runtimes.

\subsection{Matrix Representation of Super-Operators}
\label{sub:matrix_representation_of_quantum_operations}

Recall from Lemma \ref{lem-structural} that the denotational semantics of a quantum program is a (completely positive) super-operator mapping partial density operators (i.e. matrices in the finite-dimensional case considered in this paper) to themselves.  For the convenience of the reader, in this subsection let us briefly review some necessary mathematical tools for computing super-operators from \cite{Ying16} (see Section 5.1.2 there). 

Let us first show that an $d\times d$ matrix can be encoded as an $d^2$-dimensional vector using a maximally entangled state.
 Assume that $\{ \ket{j} \}_j$ is an orthonormal basis of $\HH$. Then we write:  
 $\ket{\Psi} = \sum_j \ket{j j}$ for the (unnormalised) maximally entangled state in 
  $\HH \otimes \HH$. 
  
  \begin{lemma}
  \label{lemm_mrep0}
  For any $d \times d$ matrix $A = (A_{i,j})$, we have: 
  \begin{equation}\label{entangled}
    \begin{split}
    (A \otimes I)\ket{\Psi} = (&A_{0,0}, A_{0,1}, \dots, A_{0,(d-1)}, A_{1,0}, \dots,\\ &A_{(d-1), 0}, A_{(d-1), 1}, \dots ,A_{(d-1), (d-1)}). 
\end{split}\end{equation}
\end{lemma}

Note that the left-hand side of Eq. (\ref{entangled}) is a $d^2$-dimensional vector.  

As is well-known that in practice, an abstract super-operator is usually hard to compute. But we can often use its 
matrix representation in the computation.
 
\begin{definition} Suppose  super-operator $\EE$ on a Hilbert space $\HH$ has the Kraus operator-sum representation: 
  \begin{equation}\label{Kraus}
    \EE(\rho) = \sum_m M_m \rho M_m^\dagger
  \end{equation}
  for all $\rho \in \DD(\HH)$. Then its matrix representation is defined as the following operator on $\HH\otimes\HH$: 
  \begin{equation*}
    M_{\EE} = \sum_m M_m\otimes M_m^*
  \end{equation*}
  where $M_m^*$ is the conjugate of $M_m$, and $\otimes$ stands for tensor product. 
\end{definition}

Note that when the dimension $d=\dim\HH <\infty$, each $M_m$ in Eq. (\ref{Kraus}) is a $d\times d$ matrix, and thus $M_\EE$ is a $d^2\times d^2$ matrix. 

The next lemma gives a close connection between a super-operator $\EE$ and its matrix representation through the maximal entanglement. 
\begin{lemma}
  \label{lemm_mrep1}
  For any $d \times d$ matrix $A$, we have:  
  \begin{equation}\label{super-transfer}
    (\EE(A)\otimes I)\ket{\Psi} = M_{\EE} (A \otimes I) \ket{\Psi}
  \end{equation}
\end{lemma}

We observe that in the left-hand side of Eq. (\ref{super-transfer}) super-operator $\EE$ is applied to operator $A$, but all operations on the right-hand side are matrix multiplications. As we will see below, a combination of the above two lemmas provides us with an effective way of manipulating super-operators in computing the expected runtime of quantum programs.

\subsection{Computing the Expected Runtimes of Loops}
\label{sub:computing_expected_runtime_observable_of_while_statement}

We see from Definition \ref{def_ert} that computing
runtime observable $\ert[S]$ is difficult only when $S$ contains \textbf{while}-loop. So, this subsection is devoted to develop a method for computing the following runtime observable of loop: $$\ert[\while\ M[\oq]=1\ \wdo\ S\ \wod]$$
using the matrix representation introduced in the previous subsection. 

\begin{sloppypar}
For simplicity, we write $\while$ for the loop $$\while\ M[\oq]=1\ \wdo\ S\ \wod.$$ Assume that $\ert[S]$ is given. 
First, we notice that computing $\ert[\while]$ directly using Definition \ref{def_ert} is very difficult because (possibly infinitely) many iterations of super-operators are involved there. 
\end{sloppypar}
 
However, this difficulty can be circumvented with the matrix representations of these super-operators. More precisely, 
repeatedly using Lemma \ref{lemm_mrep1}, we obtain: 
\begin{align}
   & (\ert[\while] \otimes I) \ket{\Psi}\\
  = &\label{inf-series0}  \left(\sum_{k = 0}^{\infty} \left(\EE_{P} \circ (\EE_1^* \circ \sem{S}^*)^{k}\right)(I + \EE_1^*(\ert[S]))  \otimes I\right) \ket{\Psi} \\
  = &\label{inf-series} \left( \sum_{k = 0}^\infty M_{P}(M_{\EE_1^*}M_{\sem{S}^*})^k \otimes I \right)
  \cdot \left( (I + \EE_1^*(\ert[S])) \otimes I \right) \ket{\Psi}
\end{align}
where as in Definition \ref{def_ert}, $P$ is the projection onto the almost surely terminating subspace $V_\while$ of loop $\while$, and   
\begin{align*}
  \EE_{P}(\rho) & = P \rho P^\dagger,\qquad  \EE_1(\rho) = M_1 \rho M_1^\dagger
\end{align*} for all $\rho$, $M_{P}  = P \otimes P^*$ is the matrix representation of $\EE_{P}$, and $\sem{S}^*$ is the dual of super-operator $\sem{S}$ (the denotational semantics of $S$). 
Note that the super-operators in the infinite series of (\ref{inf-series0}) are all transferred to their matrix representations in (\ref{inf-series}).

Next we compute the infinite series of matrices in (\ref{inf-series}).
Let us introduce matrix: $$R=M_{\EE_1^*}M_{\sem{S}^*}.$$ Then this infinite series can be simply written as: 
\begin{equation}\label{mat_rep_series}
\sum_{k=0}^\infty M_{P}R^k.
\end{equation}

  We see from Lemma \ref{ert_while_converge} that (\ref{mat_rep_series}) always converges. If the program $\while$ almost surely terminates on all initial states (i.e., $P = I$), then (\ref{mat_rep_series}) can be computed directly by
  \begin{equation}\label{mat_rep_series_finite}
  \sum_{k=0}^\infty M_P R^k = \sum_{k=0}^\infty R^k = (I \otimes I - R)^{-1}. 
  \end{equation}
  In the case that $P \neq I$, the series (\ref{mat_rep_series}) can be computed using the Jordan decomposition-based technique introduced in \cite{Ying13}. 

Suppose that the Jordan decomposition of 
$R$ is $$R = A J(R) A^{-1}$$ where $J(R)$ is the Jordan normal form of $R$ such 
that
\begin{equation*}
  J(R) = \bigoplus_{i = 1}^l J_{k_{i}}(\lambda_i) 
  = \diag(J_{k_1}(\lambda_1), J_{k_2}(\lambda_2), \dots, J_{k_l}(\lambda_l))
\end{equation*}
where $J_{k_i}(\lambda_i)$ is a $k_i \times k_i$-Jordan block of eigenvalue $\lambda_i$.
Since all super-operators considered in this paper do not increase the trace: $\mathit{tr}(\EE(\rho))\leq\mathit{tr}(\rho)$ for partial density operators $\rho$, we have:  
\begin{lemma}[cf. Lemma 4.1 in \cite{Ying13}] \quad
  \label{lemm_eig_le_1}
  \begin{enumerate}
    \item The eigenvalues satisfy: $\abs{\lambda_i} \leq 1$ for $1 \leq i \leq l$.
    \item If $\abs{\lambda_i} = 1$, then the dimension of the $i$th Jordan block is $k_i = 1$.
  \end{enumerate}
\end{lemma}


It is known from matrix analysis \cite{horn2012matrix} that whenever some eigenvalue $\lambda_i$ has module $1$, $\sum_{k=0}^\infty R^k$ will be divergent. Fortunately, we can remove these eigenvalues by changing $R$ to $$N=AJ(N)A^{-1}$$ where the Jordan normal form $J(R)$ of $R$ is replaced by $$J(N)=\diag(J^\prime_{k_1}(\lambda_1), J^\prime_{k_2}(\lambda_2), \dots, J^\prime_{k_l}(\lambda_l)),$$ $$J^\prime_{k_i}(\lambda_i)=\begin{cases}0\ &{\rm if}\ |\lambda_i|=1,\\ J_{k_i}(\lambda_i) &{\rm otherwise},\end{cases}$$
that is, $J^\prime_{k_i}(\lambda_i)$ is the same as the Jordan block $J_{k_i}(\lambda_i)$ of $J(R)$ when the module of its eigenvalue is less than $1$, but whenever eigenvalue $\lambda_i$ has module $1$, then the corresponding $1$-dimensional Jordan block is simply replaced by $0$. The following lemma guarantees that such a modification is feasible: 

\begin{lemma}[cf. Lemma 4.2 in \cite{Ying13}]
  \label{lemm_reduce_jnf}
  $M_{P} R^k = M_{P} N^k$ for all integers $k \geq 0$.
\end{lemma}


As a result of the above lemma, the infinite series of matrices in equation (\ref{inf-series}) can be computed as follows: 
\begin{equation}\label{inf-series2}
  \sum_{k=0}^\infty M_{P} R^k 
  = M_{P} \sum_{k=0}^\infty N^k
  = M_{P}(I\otimes I-N)^{-1}.
\end{equation}
Plugging (\ref{inf-series2}) into (\ref{inf-series}), we obtain:
\begin{equation}\label{fin-formula}\begin{split}
   & (\ert[\while]\otimes I) \ket{\Psi} \\
  & =  \left( \sum_{k=0}^\infty M_{P} R^k \right)
  \cdot \left( (I + \EE_1^*(\ert[S])) \otimes I \right) \ket{\Psi}\\
  &=  M_{P}(I\otimes I-N)^{-1} 
  \cdot \left( (I + \EE_1^*(\ert[S])) \otimes I \right) \ket{\Psi}. 
\end{split}\end{equation}
Since $\ert[S]$ is assumed to be given, $(\ert[\while]\otimes I) \ket{\Psi}$ is computed now. Furthermore, using Lemma \ref{lemm_mrep0}, $\ert[\while]$ can be retrieved from $(\ert[\while]\otimes I) \ket{\Psi}$.  

{\vskip 3pt} 

To conclude this section, we remark that for an arbitrary quantum program $S$, its expected runtime $\ERT[S]$ can be inductively computed by the above procedure combined with Definition \ref{def_ert}.
Furthermore, by Theorems \ref{main_th1} and \ref{main_th2}, either $\ERT[S](\rho)=\infty$, or it can be computed as $\ERT[S](\rho) = \tr (\ert[S] \cdot \rho).$

\section{Case Studies}
\label{sec:case_studies}

In this section, we present two more sophisticated examples to show the power of our method for computing the expected runtimes of quantum programs developed in the last section. 

\subsection{Quantum Bernoulli Factory}
\label{sub:quantum_bernoulli_factory}
Classical Bernoulli Factory (CBF) is an algorithm for generating random numbers. More precisely, 
it simulates a new coin that has probability $f(p)$ of heads given a coin
with unknown probability $p$ of heads, where $f:[0,1] \to [0,1]$ is a function.

Quantum Bernoulli Factory  (QBF) was proposed in \cite{DJR15} as a quantum counterpart of CBF. Comparing to CBF, QBF can utilize a quantum coin like 
$\ket{p}  = \sqrt{p} \ket{0} + \sqrt{1-p} \ket{1}.$
It was proved in \cite{DJR15} that QBF can simulate a strictly
larger class of function $f$ than CBF. 

{\vskip 3pt}

\textbf{Quantum Program QBF}: An example that QBF can simulate
but CBA cannot is:
\begin{equation*}
  f(p) = 1 - \abs{2p-1} = 
  \begin{cases}
    2p & p \in [0, 1/2]; \\
    2(1-p) & p \in (1/2, 1]. 
  \end{cases}
\end{equation*}
The key of simulating $f$ is to simulate
$ f'(p) = (2p-1)^2 $.
To this end, we construct a quantum coin
  $$\ket{f'(p)} = (2p-1) \ket{0} + 2\sqrt{p(1-p)}\ket{1}$$
using the following program
\begin{equation}\label{QBF}\begin{split}
  {\rm QBF} \equiv\ & q_1 := \ket{1}; q_2 := \ket{1}; \ 
   \while\ M[q_2]=1\ \wdo\ S\ \wod
\end{split}\end{equation}
where $M= \{ \ket{0}\bra{0} , \ket{1}\bra{1}\}$ is the measurement on a qubit in the computational basis, and the loop body is 
  \begin{equation}\label{QBF-i}S  \equiv\ q_1:=\ket{p}; q_2 := \ket{p}; q_1,q_2 := U[q_1,q_2]\end{equation}
  with unitary transformation: $$U = 
  \begin{pmatrix}
    \frac{1}{\sqrt{2}} & 0 & 0 & -\frac{1}{\sqrt{2}} \\
    \frac{1}{\sqrt{2}} & 0 & 0 & \frac{1}{\sqrt{2}} \\
    0 & \frac{1}{\sqrt{2}} & \frac{1}{\sqrt{2}} & 0 \\
    0 & \frac{1}{\sqrt{2}} & -\frac{1}{\sqrt{2}} & 0 \\
  \end{pmatrix}.$$

Initialisation in $\ket{1}$ and $\ket{p}$ can be realised as follows:   
\begin{align*}
  q_j & := \ket{1} \equiv q_j := \ket{0}; q_j := X[q_j] \\
  q_j & := \ket{p} \equiv q_j := \ket{0}; q_j := U_p[q_j]
\end{align*}
where $U_p$ is a unitary operator such that $$U_p \ket{0} = \sqrt{p} \ket{0} + \sqrt{1-p} \ket{1}.$$
It can be shown by a simple calculation that after every execution of $S$, the state of $q_1q_2$ is
  $$ \ket{q_1q_2} = \frac{1}{\sqrt{2}} (2p-1) \ket{00} + \frac{1}{\sqrt{2}} \ket{01} + \sqrt{2p(1-p)} \ket{10}.$$
 Thus, whenever the loop 
terminates, the state of $q_1$  is
  $$\ket{q_1} = (2p-1) \ket{0} + 2 \sqrt{p(1-p)} \ket{1},$$
which is exactly $\ket{f'(p)}$. 

Using the techniques in \cite{Ying16}, it can be verified that program QBF is almost surely terminating. 

{\vskip 3pt}

\textbf{Expected Runtime of QBF}: Now we further compute the expected runtime of program QBF using the method developed in the previous section. 
We first compute the expected runtime of the \textbf{while}-loop in QBF. Let
\begin{equation*}
  W \equiv \while\ M[q_2] = 1\ \wdo\ S\ \wod,
\end{equation*}
a direct application of formula (\ref{fin-formula}) yields:  
\begin{align*}
  (\ert[W] \otimes I) &\ket{\Psi} 
   = (I \otimes I - M_{\EE_{1}^*} M_{\sem{S}^*})^{-1}(I \otimes I) \ket{\Psi} \\
  & + (I \otimes I - M_{\EE_{1}^*} M_{\sem{S}^*})^{-1}(\EE_{1}^*(\ert[S]) \otimes I) \ket{\Psi}
\end{align*}
The left-hand side of the above equation is the vector representation of $\ert[W]$ (see Lemma \ref{lemm_mrep0}). The first term in the right-hand side is the expected time that measurement $M$ takes, and the second term is
the expected runtime that the loop body takes. 

It is easy to see that $\ert[S] = 5 \cdot I$ by the definition of $\ert$. Therefore, the runtime observable is:  
\begin{align*}
  \ert[W] &= 
  \begin{pmatrix}
    1 & 0 & 0 & 0 \\
    0 & 3 & 0 & 0 \\
    0 & 0 & 1 & 0 \\
    0 & 0 & 0 & 3 
  \end{pmatrix}
  +
  \begin{pmatrix}
    0 & 0 & 0 & 0 \\
    0 & 10 & 0 & 0 \\
    0 & 0 & 0 & 0 \\
    0 & 0 & 0 & 10 
  \end{pmatrix} \\
  &=
  \begin{pmatrix}
    1 & 0 & 0 & 0 \\
    0 & 13 & 0 & 0 \\
    0 & 0 & 1 & 0 \\
    0 & 0 & 0 & 13 
  \end{pmatrix}
\end{align*}
Hence, by Definition \ref{def_ert} we obtain the runtime observable of QBF:
  $\ert[\text{QBF}] = 17 \cdot I,$ where $I$ is the identity operator. Thus, the expected runtime of QBF is $\ERT(QBF)(\rho)=17$ for all density operators.  
It is interesting to see that the expected runtime is independent of the probabilistic parameter $p$.

\subsection{Quantum Random Walk}
\label{sub:quantum_random_walk}

We were able to show in Example \ref{ex-rw-EXT} by Definition \ref{def_EXT}  that the expected runtime of a simple quantum random walk, namely Hadamard walk, initialised in direction $L$ and position $1$, is $\infty$. In this subsection, we consider a more complicated quantum random walk, a quantum walk on an $(n+1)$-circle with absorbing boundaries at positions $0$ and $n$, of which the expected runtime is hard to compute using Definition \ref{def_EXT}  directly. Instead, we will compute it using the method introduced in the previous section. 

{\vskip 3pt}

\textbf{Quantum Program QW}: The quantum coin is the same as before, with $\HH_q$ spanned by the orthonormal basis $\{ \ket{L}, \ket{R} \}$ as its state Hilbert space. But the position space is an $n$-dimensional Hilbert space $\HH_p$ with orthonormal basis
$\{ \ket{0}, \ket{1},$ $\dots, \ket{n} \}$, where basis state $\ket{i}$ is used to denote position $i$ on the circle. The state space of quantum walk
is $\HH \stackrel{\triangle}{=} \HH_q \otimes \HH_p$. Each step of the quantum walk consists of: 
\begin{enumerate}
  \item Measure the position of the system to see whether it is the absorbing boundary $0$ or $n$. If it is the case, the walk terminates; otherwise, it continues. Mathematically,  the measurement is described as $M = \{ M_0, M_1 \}$, where
    $$
      M_0 = \ket{0}_p\bra{0} + \ket{n}_p\bra{n}, \quad M_1 = I_p - M_0 = \sum_{k = 1}^{n-1} \ket{k}_p\bra{k}.
    $$
  \item Toss the coin by applying an operator $T$ on $\HH_q$: 
    \begin{equation*}
      T = 
      \begin{pmatrix}
        a & \ b^* e^{i\phi}\\
        b & \ -a^* e^{i\phi}
      \end{pmatrix}
    \end{equation*} where $a,b$ are complex numbers satisfying the normalisation condition: $\abs{a}^2 + \abs{b}^2 = 1.$ Note that the coin tossing operator here is a general $2\times 2$ unitary operator rather than the Hadamard gate $H$.  
  \item Shift the position to the left or right according to the state of  
    coin. The shift operator is given as 
    \begin{equation*}
      S = \sum_{i = 0}^{n} \ket{L} \bra{L} \otimes \ket{i \ominus 1} \bra{i}
      + \sum_{i = 0}^{n} \ket{R} \bra{R} \otimes \ket{i \oplus 1} \bra{i}
    \end{equation*}
    where $\oplus$ and $\ominus$ are addition and subtraction modulo $n + 1$. Note that addition and subtraction modulo $n+1$ are used here because the walk is on an $(n+1)$-circle.   
\end{enumerate}

The above process can be formally described as the following quantum loop:
\begin{equation}\label{n-QW}
  QW_n \equiv \while\ M[p] = 1\ \wdo\ q := T[q]; q,p := S[q, p]\ \wod
\end{equation}

{\vskip 3pt}

\textbf{Expected Runtime of QW}: The expected runtime of $QW_n$ has been an open problem since \cite{Ambainis2001}, and it was proved in \cite{Ying13} to be $n$ for a special initial state with $n < 30$.  
Here, we compute $\ERT(QW_n)$ for the general case using the method developed in the previous section. 
To this end, let us write: $$QW_n' \equiv q := T[q]; q,p := S[q, p]$$ for its loop body. 
  It was shown in Theorem 4.5 of \cite{kuklinski2018absorption} that $QW_n$ almost surely terminates on all computational basis states. Based on this, we can prove that $QW_n$ almost surely terminates on all initial states by Lemma \ref{lemm_term_space}, and thus has a finite expected runtime by Theorem \ref{main_th2}.
Then by formulas  (\ref{fin-formula}) and (\ref{mat_rep_series_finite}), the runtime observable of $QW_n$ can be computed as follows:
\begin{align*}
  &(\ert[QW_n] \otimes I) \ket{\Psi} 
   = (I \otimes I - E \otimes E^*)^{-1}(I \otimes I) \ket{\Psi} \\
  & \qquad\qquad + (I \otimes I - E \otimes E^*)^{-1}(\EE_{1}^*(\ert[QW_n']) \otimes I) \ket{\Psi}
\end{align*} where $$E = M_1^\dagger (T \otimes I)^\dagger S^\dagger.$$

We are more interested in the first term in the right-hand side of
the above equation because it is actually the expected steps that the quantum random walk goes
plus one. Let $$Q_n = \sum_{k = 0}^{\infty} (\EE_1^* \circ \sem{QW'}^*)^{k}(I).$$ Then 
we have: $$(Q_n \otimes I) \ket{\Psi} = (I \otimes I - E \otimes E^*)^{-1} \ket{\Psi}.$$
Consequently, $\bra{L,k} Q_n \ket{L, k}$ and $\bra{R, k} Q_n \ket{R, k}$ are 
exactly the expected steps that the quantum random walk takes when it  
begins with states $\ket{L, k}$ and $\ket{R, k}$, respectively.

In practice, it is a bit difficult to calculate $Q_n$ if $n$ is treated as an abstract 
parameter. However, we can easily calculate $Q_n$ when $n$ is a given integer. 
Moreover, we can guess a pattern $X$ of $Q_n$ for arbitrary $n$ from 
these results of some given values of $n$. Then, it is sufficient to show that
the pattern that we guessed is exactly $Q_n$. Suppose that we have a matrix $X$ such that $$I + E X E^\dagger = X.$$ Then we obtain: 
\begin{align*}
  (X \otimes I) \ket{\Psi} 
  & = ((I + E X E^*)\otimes I) \ket{\Psi} \\
  & = \ket{\Psi} + (E \otimes E^*)(X \otimes I)\ket{\Psi}
\end{align*}
by Lemma \ref{lemm_mrep1}. The previous equation is equivalent to 
\begin{equation*}
  (I \otimes I - E \otimes E^*) (X \otimes I) \ket{\Psi} = \ket{\Psi}
\end{equation*}
Since this quantum random walk is almost surely terminating,
$(I \otimes I - E \otimes E^*)$ is invertible. Thus, we have:
$$(Q_n \otimes I) \ket{\Psi} = (X \otimes I) \ket{\Psi}.$$
Note that $(Q_n \otimes I) \ket{\Psi}$ is an encoding of $Q_n$ 
according to Lemma \ref{lemm_mrep0}. Then we can conclude that
$Q_n = X$ by Lemma \ref{lemm_mrep0}.  

Now we compute $Q_n$ in the following four steps:

{\vskip 3pt}

\textbf{\textit{Step 1: Reduce to real coin tossing operators}}: The entries of coin tossing operator $T$ are allowed to be complex numbers. But we can show that it is sufficient to deal with 
the case where all entries of $T$ are reals. Since $T$ is unitary, by the $Z$-$Y$ decomposition (see Theorem 4.1 in \cite{NC}), we can find reals $x, y, \alpha, \beta, \delta$ so that
\begin{equation*}
  T = e^{i\alpha}
  \begin{pmatrix}
    e^{-i(\beta + \delta)} x & e^{-i(\beta - \delta)} y \\
    e^{i(\beta - \delta)} y & -e^{i(\beta + \delta)} x
  \end{pmatrix}, x^2 + y^2 = 1.
\end{equation*}
Note that $EXE^\dagger$ is irrelevant to $\alpha$, we can assume $\alpha = 0$
here.

Consider another quantum walk that uses the coin-tossing operator: 
\begin{equation*}
  T' = 
  \begin{pmatrix}
    x & y \\
    y & -x
  \end{pmatrix} 
\end{equation*}
Note that all entries of matrix $T^\prime$ are reals. 
Let $$E' = M_1^\dagger (T' \otimes I)^\dagger S^\dagger,$$ $$(Q_n' \otimes I) \ket{\Psi} = (I \otimes I - E' \otimes {E'}^\dagger)^{-1}\ket{\Psi}.$$ Then we have: 

\begin{lemma}
  \label{lemm_real_is_enough}
  $Q_n$ and $Q_n'$ are related by a unitary operator $P$: 
  \begin{equation*}
    Q_n = P^\dagger Q_n' P
  \end{equation*}
  where: 
$$    P = 
    \begin{pmatrix}
      P_L & \zeromat \\
      \zeromat & P_R
    \end{pmatrix}, $$
    $$\quad P_L = \sum_{k = 0}^{n} e^{-i[(\beta + \delta)k + 2 \delta]} \ket{k}\bra{k},\  
    P_R = \sum_{k = 0}^{n} e^{-i(\beta + \delta) k} \ket{k} \bra{k}.$$
\end{lemma}

With the above lemma, we only need to consider the case where 
$\phi = 0$ and $a, b$ are both reals.

{\vskip 3pt}

\textbf{\textit{Step 2: Find the pattern of $Q_n$}}: To this end, 
we first compute $Q_n$ with the Hadamard coin-tossing operator $T = H$ for some fixed $n$'s. Our results shows  
that $Q_n$ has the form:
\begin{equation}\label{pattern}
  Q_n = 
  \begin{pmatrix}
    A_n & B_n^\dagger \\
    B_n & C_n
  \end{pmatrix}
\end{equation}
where $A_n$, $B_n$ and $C_n$ are all $(n+1) \times (n+1)$ matrices. Moreover, 
we see that the non-zero entries $(A_n)_{i,j}$ and $(C_n)_{i, j}$ 
in $A_n$ and $C_n$ are all quadratic polynomials of $i$ and $j$, 
and the non-zero entries $(B_n)_{i, j}$ in  $B_n$ 
are linear in $i$ and $j$. 

{\vskip 3pt}

\textbf{\textit{Step 3: Solution of $Q_n$}}: Finally, we can present a solution of $Q_n$ for a general coin tossing operator $T$ with real entries. We use pattern (\ref{pattern}), so what we need to compute are matrices $A_n$, $B_n$, and $C_n$.
Let
\begin{align*}
  f_n(j, k) & = (-1)^{\frac{j - k}{2}} \cdot \frac{b^2}{a^2} \cdot (j \bmod n) \cdot (n - 1 - k), \\
  h_n(j, k) & = (-1)^{\frac{j - k}{2}} \cdot \frac{b}{a} \cdot (j + k - n).
\end{align*}
 The solutions of $A_n, B_n$ and $C_n$ can be given as follows:
\begin{align*}
  (A_n)_{j, k} & = 
  \begin{cases}
    \begin{aligned}&f_n(j, j) + 1 \\
    &+(j \bmod n)\end{aligned} & j = k, \\
    f_n(j, k) & 0 < j < k < n \text{ and } 2 \mid (j - k), \\
    f_n(k, j) & 0 < k < j < n \text{ and } 2 \mid (j - k), \\
    0 & \text{otherwise},
  \end{cases} \\
  (B_n)_{j, k} & =
  \begin{cases}
    h_n(j, k) & 0 < j \leq k < n \text{ and } 2 \mid (j - k), \\
    0 & \text{otherwise},
  \end{cases}\\
  (C_n)_{j, k} & = (A_n)_{n - j, n - k}.
\end{align*}

\textbf{\textit{Step 4: Verification of $Q_n$}}: The correctness of the above solutions can be verified by checking the following equation:
\begin{equation*}\begin{aligned}
  &\begin{pmatrix}
    A_n & B_n^\dagger \\
    B_n & C_n
  \end{pmatrix}
  = I + E 
  \begin{pmatrix}
    A_n & B_n^\dagger \\
    B_n & C_n
  \end{pmatrix}
  E^\dagger\\ 
  & \qquad = 
  I + 
  \begin{pmatrix}
    M_1 & 0 \\
    0 & M_1
  \end{pmatrix}
  \begin{pmatrix}
    a S_L^\dagger & b S_L \\
    b S_L^\dagger & -a S_L
  \end{pmatrix}
  \begin{pmatrix}
    A_n & B_n^\dagger \\
    B_n & C_n
  \end{pmatrix} \\
  &\qquad\qquad\qquad \cdot
  \begin{pmatrix}
    a S_L & b S_L \\
    b S_L^\dagger & -a S_L^\dagger
  \end{pmatrix}
  \begin{pmatrix}
    M_1 & 0 \\
    0 & M_1
  \end{pmatrix},\end{aligned}
\end{equation*}
where $S_L = \sum_{i = 0}^{n} \ket{i \ominus 1} \bra{i}.$
This equation can be divided into four equations for the sub-matrices. Each of them can be checked directly by matrix calculation. As an example, we check
the equation for the top left submatrix, which is
\begin{multline*}
  A_n = I_n + M_1 (a^2 S_L^\dagger A_n S_L + b^2 S_L C_n S_L^\dagger\\ 
  + ab S_L B_n S_L + ab S_L^\dagger B_n^\dagger S_L^\dagger)M_1
\end{multline*}
According to the construction of $B_n$, there are two cases which 
depend on whether the element of $A_n$ is on the diagonal. These two cases can be
reduced to the equations:
\begin{equation} \label{fun-1}
  \begin{split}
    (A_n)_{j, j} & = 1 + a^2 \cdot (A_n)_{j - 1, j - 1} + b^2 \cdot (A_n)_{n-(j+1), n-(j+1)}, \\
    (A_n)_{j, k} & = a^2 \cdot (A_n)_{j - 1, k - 1} + b^2 \cdot (C_n)_{j + 1, k + 1}  - ab \cdot (B_n)_{j + 1, k - 1}
  \end{split}
\end{equation}
for $0 < j < k < n$.
Both of the above equations can be checked by simple calculation. As a result, we have:
\begin{proposition}The expected steps of $QW_n$ starting from state $\ket{L, k}$ and $\ket{R,k}$ are:
\begin{align*}
  \bra{L,k} Q_n \ket{L, k} &= f_n(k, k) + (k \bmod n) + 1,\\ 
  \bra{R, k} Q_n \ket{R, k}&= f_n(n - k, n - k) + ((n - k) \bmod n) + 1
\end{align*} respectively. 
More generally, if starts from $$|\Psi\rangle=\sum_{k=0}^{n}(\alpha_k |L,k\rangle+\beta_k |R,k\rangle),$$ the expected step of $Q_n$ is: 
\begin{align*}
  & \langle\Psi |Q_n|\Psi\rangle 
   =  \sum_{j = 0}^{n} \Bigg[
  \big(f_{n}(j, j) + (j \bmod n) + 1\big) \\
  & \cdot \Big(\alpha_j^*\alpha_j + \beta_{n - j}^*\beta_{n - j}\Big) \Bigg] +  \sum_{j = 1}^{n - 1} \Bigg[h_n(j, j) (\alpha_j^*\beta_j + \beta_j^*\alpha_j)\Bigg] \\
   & + \underset{\begin{subarray}{c} 0 < j < k < n, 2 \mid (k - j)\end{subarray}} {\sum} \Bigg[ h_n(j, k) \Big(\beta_j^*\alpha_k + \alpha_k^*\beta_j\Big) \\
  & + f_n(j, k)\Big(\alpha_j^*\alpha_k + \alpha_k^*\alpha_j + \beta_{n - j}^*\beta_{n - k} + \beta_{n - k}^*\beta_{n - j}\Big) 
\Bigg]
\end{align*}
\end{proposition}

It deserves to mention
that the first term in Eq. (\ref{fun-1}) coincides with the expected runtime of classical 
random walk on an $n$-circle which moves to $k - 1$ with probability 
$a^2$ and moves to $n - (k + 1)$ with probability $b^2$ from position 
$0 < k < n$.  For example, the expected steps taken by the quantum random walk on 
a $6$-circle starting from state $\ket{L}\ket{k}$ are equal to the expected 
runtime of the classical random walk in Figure \ref{fig_cw} starting from $k$ 
and terminating at $t$.
\begin{figure}[htbp]
  \centering
  \begin{tikzpicture}[shorten >=1pt,node distance=1.2cm,on grid]
    \node[state] (0) [scale=1] {$0$};
    \node[state] (t) [right =of 0, xshift=1cm, scale=1] {$t$};
    \node[state] (1) [below left=of 0, xshift=-1.2cm, yshift=-0.7cm, scale=1] {$1$};
    \node[state] (4) [below right=of 0, xshift=1.2cm, yshift=-0.7cm, scale=1] {$4$};
    \node[state] (2) [below=of 1, xshift=0.8cm, yshift=-1.2cm, scale=1] {$2$};
    \node[state] (3) [below=of 4, yshift=-1.2cm, xshift=-0.8cm, scale=1] {$3$};
    \path[->] 
      (0) edge node [above] {$1$} (t)
      (1) edge node [above left] {$a^2$} (0)
          edge [] node [below left] {$b^2$} (3)
      (2) edge node [left] {$a^2$} (1)
          edge [loop below] node [left] {$b^2$} (2)
      (3) edge node [below] {$a^2$} (2)
          edge [bend right] node [above right] {$b^2$} (1)
      (4) edge node [right] {$a^2$} (3)
          edge node [above right] {$b^2$} (0);
  \end{tikzpicture}
  \caption{Classical random walk related to the quantum random walk on $6$-circle}
  \Description{Classical random walk related to the quantum random walk on $6$-circle}
  \label{fig_cw}
\end{figure}

\section{Conclusion}
In this paper, we defined the expected runtimes of quantum programs as a generalisation of quantum weakest precondition \cite{d2006quantum}.
A representation of the expected runtimes as a quantum observable was presented. 
This representation gives a physical interpretation of the notion of expected runtime. Based on it, we develop an effective method for computing the expected runtimes of quantum programs in finite-dimensional Hilbert spaces using the mathematical tool of the matrix representation of super-operators.
We demonstrated the power of our computational method through several case studies, including the expected runtime of quantum Bernoulli factory --- a quantum algorithm for generating random numbers; in particular, our method is able to compute the expected runtime of quantum walk on an $n$-circle, for arbitrary $n$, an arbitrary quantum coin and an arbitrary initial state, and thus solve an open problem.

The basic idea of this paper came from recent work on the corresponding problem for probabilistic programs \cite{kaminski2016,kaminski2018,Ngo2018}, but the computational method presented in this paper is quite different from there. 
Except for a weakest precondition calculus, a set of proof rules for reasoning about the expected runtime of probabilistic programs was also presented in \cite{kaminski2016,kaminski2018}. It is also possible to develop some similar proof rules for quantum programs by extending quantum Hoare logic \cite{Ying11}.

Our approach to computing the expected runtime of quantum programs is limited to the case of finite-dimensional Hilbert spaces and thus cannot deal with the quantum integer type. Nevertheless, most of the existing quantum algorithms have been designed in the finite-dimensional case, and our results can be applied to them. On the other hand, the infinite-dimensional case is certainly an interesting (and challenging, we believe) topic for future research.

We saw in Subsection \ref{sub:quantum_random_walk} that the computation of the expected runtime of a quantum program can be much more involved than that of a probabilistic program. So, one of the most important topics for future research is  efficient symbolic automation of our computational method; more specifically, for example, how to combine the quantum weakest precondition-style reasoning developed in this paper with the automatic amortised resource analysis (AARA) \cite{Ngo2018,carbonneaux2017automated,celiku2005compositional,hoffmann2015type,hofmann2015multivariate} for computing the expected runtimes of more complicated  quantum programs. 

\begin{acks}
This work was supported by the National Key R\&D Program of China (Grant No: 2018YFA0306701) and the National Natural Science Foundation of China (Grant No: 61832015).
\end{acks}

\newpage

\bibliographystyle{ACM-Reference-Format}
\bibliography{ert}

\onecolumn

\begin{appendix}

{\centering\Large \textbf{Supplementary material and deferred proofs}}

\section{Basics of Quantum Theory}
\label{sec:preliminaries}


Quantum computing acts on microsystems which dominated by quantum mechanics and can be described by linear algebra over the field of complex numbers $\mathbb{C}$.
Analog to the classical logical bit with one of two possible values $0$ or $1$,  a quantum bit (aka qubit) also has two basis states which are conventionally denoted by Dirac—or ``bra–ket''—notation as $\ket{0}$ and $\ket{1}$, representing two column vectors $\ket{0}=\left(\begin{array}{c} 1 \\ 0\end{array}\right)$ and $\ket{1}=\left(\begin{array}{c} 0 \\ 1\end{array}\right)$.
Different from classical state, the pure state of a qubit can also be in a superposition of basis states, i.e., $a_0 \ket{0} + a_1 \ket{1}$, where $a_0, a_1\in\mathbb{C}$ satisfying normalization condition $\abs{a_0}^2 + \abs{a_1}^2 = 1$, for example, $$\ket{+}
\equiv\frac{1}{\sqrt{2}}(|0\rangle+ |1\rangle),\quad \ket{-}\equiv\frac{1}{\sqrt{2}}(|0\rangle- |1\rangle)$$
are also pure states. The set of all states (without normalization condition) equipped with inner product formulate the two-dimensional Hilbert space $\mathcal{H}_2$. 
Similarly, a quantum integer (aka quint) analog to classical integer has infinite basis states $\ket{0}, \ket{1}, \ket{2}, \dots$ and a pure state of quint has form $a_0\ket{0} + a_1\ket{1} +  a_2\ket{2} + \cdots$ with $\sum_{j=0}^\infty \abs{a_j}^2 = 1$. We use $\mathcal{H}_\infty$ to denote the infinite-dimensional Hilbert space.

The evolution of a quantum system is described by linear transformations on Hilbert space. For closed systems, the evolution is deterministic and may be represented by a unitary transformation, i.e., a unitary matrix $U$. For example, the frequently used Hadamard gate $H$ has the matrix form $\frac{1}{\sqrt{2}}\left(\begin{array}{cc}1 & 1\\ 1 & -1\end{array}\right)$, and if we apple it to a qubit state $\ket{0}$, the final state we obtained is $$H\ket{0} = \frac{1}{\sqrt{2}}\left(\begin{array}{cc}1 & 1\\ 1 & -1\end{array}\right)\left(\begin{array}{c} 1 \\ 0\end{array}\right) = \frac{1}{\sqrt{2}}\left(\begin{array}{c} 1 \\ 1\end{array}\right) =  \frac{1}{\sqrt{2}}(|0\rangle+ |1\rangle) = \ket{+}.$$ 

Another important evolution is quantum measurement, which is often used to extract classical information from quantum states. Mathematically, a measurement is described by a collection of matrices $M=\{ M_j \}$ with completeness equation $\sum_j M_j^\dagger M_j = I$ where $I$ is the identity matrix. Performing a measurement $M$ on the state $\ket{s}$ yields an classical outcome $j$ with probability $\norm{M_j \ket{s}}^2$ and the state changes to $M_j \ket{s} / \norm{M_j \ket{s}}$ correspondingly. Consider the measurement in the computational basis, i.e., $M = \{M_0 = \ket{0}\bra{0}, M_1 = \ket{1}\bra{1}\}$ (the matrix form of $M$ is $\left\{\left(\begin{array}{cc}1 & 0\\ 0 & 0\end{array}\right), \left(\begin{array}{cc}0 & 0\\ 0 & 1\end{array}\right)\right\}$), if we measure state $\ket{+}$ using $M$, we have probability $\norm{M_0 \ket{+}} = \frac{1}{2}$ to obtain classical outcome $0$ and the state changes to $M_0 \ket{+} / \norm{M_0 \ket{+}} = \ket{0}$, and similarly we also have the rest probability $\norm{M_1 \ket{+}} = \frac{1}{2}$ to get outcome $1$ with the post-measurement state $\ket{1}$.
It should be point that we can not directly read the state of a quantum program, 
This means that the access of information of quantum state often results in 
changes in the state. 

Quantum measurement brings probabilism to quantum programs, i.e., the quantum state may not completely known. A conventional way to handle this is to use density operator language, where quantum states are described by matrices instead of vectors.  
Consider a quantum system which is in one of the states $\ket{s_j}$ with respective probability $p_j$, its density operator is defined by $$\rho \equiv \sum_j p_j \ket{s_j}\bra{s_j}.$$ For example, after we perform $M = \{M_0 = \ket{0}\bra{0}, M_1 = \ket{1}\bra{1}\}$ on state $\ket{+}$, the post-measurement state is either $\ket{0}$ or $\ket{1}$ with same probability $\frac{1}{2}$, and if we ignore the classical outcome and only focus on the quantum state, such post-measurement state should be described by density operator $$\frac{1}{2}\ket{0}\bra{0}+\frac{1}{2}\ket{1}\bra{1} = \frac{1}{2}\left(\begin{array}{cc}1 & 0\\ 0 & 1\end{array}\right).$$ 
Mathematically, a density operator refers to positive semi-definite matrix $\rho$ with trace (i.e., the sum of diagonal elements) $\tr(\rho) = 1$, and for those with trace less equal 1, we call it partial density operator. Correspondingly, in density operator language, a unitary transformation $U$ maps density operator $\rho$ to $U\rho U^\dag$ where $U^\dag$ is the conjugate-transpose of $U$, and a measurement $M=\{ M_j \}$ applied on $\rho$ yields outcome $j$ with probability $\tr(M_j\rho M_j^\dag)$ and post-measurement state $$M_j\rho M_j^\dag / \tr(M_j\rho M_j^\dag).$$ We may also describe the evolution of an open system deterministically using density operator language, which is known as quantum operation. A well-known 
representation of quantum operations is the operator-sum representation: 
for a quantum operation $\EE$, there exists matrices $\{ E_j \}$ such that 
$I - \sum_j E_j^\dagger E_j$ is positive and for all density operator 
$\rho$ as input, the output state $\EE(\rho)$ after performing $\EE$ is $$\EE(\rho) = \sum_j E_j \rho E_j^\dagger.$$ 

Finally, it is frequently to consider the composite quantum systems made up of two (or more) distinct physical systems. Suppose we have two qubits $q_1$ and $q_2$ and the states of each system is $\ket{\psi_1}$ and $\ket{\psi_2}$ respectively, then the joint state of the total system is described by tensor product $\ket{\psi_1}\otimes\ket{\psi_2}$, and may be written in abbreviation  $\ket{\psi_1}\ket{\psi_2}$ or $\ket{\psi_1, \psi_2}$. Together with the superposition principle, there exists state such as $\ket{\Phi} = \frac{1}{2}(\ket{0}\ket{0}+\ket{1}\ket{1})$ which cannot be written as the tensor product of two single-qubit states $\ket{s_1}$ and $\ket{s_2}$, and we call such states (with property that it cannot be written as a product of states of its component systems) entangled states.


\section{Operational Semantics of Quantum While-Programs}
\label{sub:operational_semantics_of_quantum_while_programs}

The denotational semantics in Lemma \ref{lem-structural} and the $\ERT$ in Definition \ref{def_EXT} can be derived from the operational semantics below.

For quantum program $S$, a quantum configuration is a pair $C = \langle S, \rho \rangle$ where $S$ is
a program or the termination symbol $\downarrow$, and $\rho \in \DD(\HH_S)$
denotes the state of the quantum variables in $S$.
The operational semantics of quantum programs can be defined in a way similar
to classical programs:

\begin{definition}[Operational Semantics \cite{Ying11}]
  \label{def_operational_semantics}
  Let $N$ be the largest number of outcomes of the measurement in the program.
  Let $A = \{\Sk, \In, \Un, \Wh_0, \Wh_1\} \cup \{ \Br_m \mid 0 \leq m < N\}$ be an alphabet which represent the set of program actions.
  The operational semantics of quantum \textbf{while}-programs is a transition
  relation $\to$ between quantum configurations defined by the transition rules
  in Figure \ref{fig 3.1}.

\begin{figure}[h]\centering
  \begin{equation*}
    \begin{split}&({\rm Sk})\ \ \langle\mathbf{skip},\rho\rangle\xrightarrow{\Sk}\langle \downarrow,\rho\rangle \\
&({\rm In})\ \ \ \langle q:=|0\rangle,\rho\rangle\xrightarrow{\In}\langle \downarrow,\rho^{q}_0\rangle\\
&({\rm UT})\ \ \langle\overline{q}:=U[\overline{q}],\rho\rangle\xrightarrow{\Un}\langle \downarrow,U\rho U^{\dag}\rangle \\
&({\rm SC})\ \ \ \frac{\langle S_1,\rho\rangle\xrightarrow{a}\langle S_1^{\prime},\rho^{\prime}\rangle} 
{\langle S_1;S_2,\rho\rangle\xrightarrow{a}\langle S_1^{\prime};S_2,\rho^\prime\rangle}\\
&({\rm IF})\ \ \ \langle\mathbf{if}\ (\square m\cdot M[\overline{q}]=m\rightarrow S_m)\ \mathbf{fi},\rho\rangle\xrightarrow{\Br_m}\langle S_m,M_m\rho M_m^{\dag}\rangle\\
&({\rm L0})\ \ \ \langle\mathbf{while}\ M[\overline{q}]=1\ \mathbf{do}\ S\ \mathbf{od},\rho\rangle\xrightarrow{\Wh_0}
\langle \skp, M_0 \rho M_0^\dagger \rangle\\
&({\rm L1})\ \ \ \langle\mathbf{while}\ M[\overline{q}]=1\ \mathbf{do}\ S\ \mathbf{od},\rho\rangle\xrightarrow{\Wh_1}
\langle S; \mathbf{while}\ M[\overline{q}]=1\ \mathbf{do}\ S\ \mathbf{od}, M_1 \rho M_1^\dagger \rangle
    \end{split}
    \end{equation*}
\caption{Transition Rules. In (In), $\rho^{q}_0=\sum_n|0\rangle_q\langle n|\rho|n\rangle_q\langle 0|$, where $\{|n\rangle\}$ is an orthonormal basis of $\HH_q$.
In (SC), we make the convention $\downarrow;S_2=S_2.$
In (IF), $m$ ranges over every possible outcome of measurement $M=\{M_m\}.$}\label{fig 3.1}
\end{figure}
 \end{definition}
 
The rules (Sk) and (SC) are the same as in classical or probabilistic programming. Other rules are determined by the basic postulates of quantum mechanics. In particular, (IF), (L0) and (L1) are essentially probabilistic, 
but we choose to present them as nondeterministic
transitions, following a convention from \cite{Selinger04}: a probabilistic transition $\langle S,\rho\rangle\stackrel{p}{\rightarrow}\langle S^\prime,\rho^\prime\rangle$ can be identified with the non-probabilistic transition   
$\langle S,\rho\rangle{\rightarrow}\langle S^\prime,\rho^{\prime\prime}\rangle$, where $\rho,\rho^\prime$ are density operators denoting the program states before and after the transition, respectively, $p$ is the probability of the transition, and $\rho^{\prime\prime}=p\rho^\prime$ is a partial density operator. Obviously, transition probability $p$ can be retrieved from partial density operator $\rho^{\prime\prime}$ as its trace: $p=\tr(\rho^{\prime\prime})$. 
For example in rule (IF), a measurement outcome $m$ occurs with probability $p_m=\tr(M_m\rho M_m^\dag)$ and in this case $\rho_m=M_m\rho M_m^\dag/p_m$ will be the state after the measurement.  
They are combined into partial operator $p_m\rho_m=M_m\rho M_m^\dag$ denoting the program state after transition. 
This convention significantly simplifies the subsequent presentation, and its reasonableness is guaranteed by the linearity in quantum mechanics.   


\begin{definition}[Denotational Semantics \cite{Ying11}]
  \label{def_denotational_semantics}
  The denotational semantics of a quantum \textbf{while}-program $S$ is the
  mapping $\sem{S} : \DD(\HH_S) \to \DD(\HH_S)$ defined by
  \begin{equation*}
    \sem{S}(\rho) = \sum_{a \in A^*, \langle S, \rho\rangle \xrightarrow{a}^* \langle \downarrow, \rho'\rangle} \rho'
  \end{equation*}
  for every $\rho \in \DD(\HH_S)$, where $\xrightarrow{\bullet}^*$ is the reflexive and transitive closure of $\xrightarrow{\bullet}$.
\end{definition}

Intuitively, the output of quantum program $S$ on input state $\rho$ is the sum of all partial states $\rho'$ at which the program terminates in a finite number of steps.

\section{Derivation of ERT from Operational Semantics}
\label{sub:derivation_of_ert_from_operational_semantics}

The expected runtime function in Definition \ref{def_EXT} can be obtained from the above operational semantics. We first define the time cost for progressing one step in a quantum program.

\begin{definition}
  The one-step time cost function $c$ is defined by
  \begin{itemize}
    \item $c(\downarrow) = 0$;
    \item $c(\skp) = 0$;
    \item $c(q := \ket{0}) = 1$; 
    \item $c(\oq := U[\oq]) = 1$; 
    \item $c(S_1; S_2) = c(S_1)$;
    \item $c(\ifb(\guard m\cdot M[\oq] = m \to S_m)\ife) = 1$;
    \item $c(\while\ M[\oq]=1\ \wdo\ S\ \wod) = 1$.
  \end{itemize}
\end{definition}

$c(S)$ denotes the time cost for program $S$ to move one step in Definition \ref{def_operational_semantics}. 

Note that the transition relation in Definition \ref{def_operational_semantics} can be split into two functions $T_1$ and $T_2$ such that
\begin{equation*}
  T_1(S, a) = 
  \begin{cases}
    S' & \exists \rho'.\ \langle S, \rho \rangle \xrightarrow{a} \langle S', \rho' \rangle \\
    \skp & \text{otherwise},
  \end{cases}
  \quad \quad \quad
  T_2(S, a, \rho) = 
  \begin{cases}
    \rho' & \exists S'.\ \langle S, \rho \rangle \xrightarrow{a} \langle S', \rho' \rangle \\
    \zeromat & \text{otherwise}.
  \end{cases}
\end{equation*}
Moreover, we can prove that $T_2$ is linear on $\rho$. Then Definition \ref{def_denotational_semantics} can be written as
\begin{equation*}
  \sem{S}(\rho) = \sum_{a \in A^*, T_1(S, a) = \downarrow} T_2(S, a, \rho).
\end{equation*}

Now, we can give an alternative definition of $\ERT$:

\begin{definition}[Alternative definition of $\ERT$]
  The expected runtime of a quantum \textbf{while}-program $S$ on $\rho \in \DD(\HH_S)$ is defined by
  \begin{equation*}
    \ERT'[S](\rho) = \sum_{a \in A^*, \langle S, \rho\rangle \xrightarrow{a}^* \langle S', \rho'\rangle} c(S')\tr(\rho') 
    = \sum_{a \in A^*} c(T_1(S, a)) \tr(T_2(S, a, \rho))
  \end{equation*}
\end{definition}

The two definitions of $\ERT$ can be related by the following Lemma:

\begin{lemma}
  For quantum program $S$ and state $\rho \in \DD(\HH_S)$, we have $\ERT'[S](\rho) = \ERT[S](\rho)$.
\end{lemma}

\begin{proof}
  We proceed by induction on the structure of $S$.
  \begin{itemize}
    \item $S \equiv \skp$ , $S \equiv \oq := U[\oq]$ and $S \equiv q := \ket{0}$: The result is straightforward.
    \item $S \equiv S_1; S_2$: With the induction hypothesis, we have:
      \begin{align*}
        \ERT'[S](\rho) & = \sum_{a \in A^*, \langle S, \rho\rangle \xrightarrow{a}^* \langle S', \rho'\rangle} c(S')\tr(\rho') \\
       & = \sum_{a \in A^*, \langle S_1, \rho\rangle \xrightarrow{a}^* \langle S', \rho'\rangle} c(S')\tr(\rho')
       + \sum_{a \in A^*, \langle S_1, \rho\rangle \xrightarrow{a}^* \langle \downarrow, \rho'\rangle} 
       \left(\sum_{a' \in A^*, \langle S_2, \rho' \rangle \xrightarrow{a'}^* \langle S', \rho'\rangle} c(S')\tr(\rho') \right) \\
       & = \ERT'[S_1](\rho)
       + \sum_{a \in A^*, \langle S_1, \rho\rangle \xrightarrow{a}^* \langle \downarrow, \rho'\rangle} 
       \left(\sum_{a' \in A^*} c(T_1(S_2, a')) \tr(T_2(S_2, a', \rho')) \right) \\
       & = \ERT'[S_1](\rho)
       + \sum_{a' \in A^*} \left(\sum_{a \in A^*, \langle S_1, \rho\rangle \xrightarrow{a}^* \langle \downarrow, \rho'\rangle}  
       c(T_1(S_2, a')) \tr(T_2(S_2, a', \rho')) \right) \\
       & = \ERT'[S_1](\rho)
       + \sum_{a' \in A^*} 
       c(T_1(S_2, a')) \tr\left(T_2(S_2, a', \sum_{a \in A^*, \langle S_1, \rho\rangle \xrightarrow{a}^* \langle \downarrow, \rho'\rangle} \rho' )\right) \\
       & = \ERT'[S_1](\rho)
       + \sum_{a' \in A^*} 
       c(T_1(S_2, a')) \tr(T_2(S_2, a', \sem{S_1}(\rho) )) \\
       & = \ERT'[S_1](\rho) + \ERT'[S_2](\sem{S_1}(\rho)) \\
       & = \ERT[S_1](\rho) + \ERT[S_2](\sem{S_1}(\rho)) \\
       & = \ERT[S](\rho).
      \end{align*}
    \item $S \equiv \ifb(\guard m\cdot M[\oq] = m \to S_m)\ife$: With the induction hypothesis, we have: 
        \begin{align*}
          \ERT'[S](\rho) & = \sum_{a \in A^*, \langle S, \rho\rangle \xrightarrow{a}^* \langle S', \rho'\rangle} c(S')\tr(\rho') \\
                         & = c(S)\tr(\rho) + \sum_m \sum_{a \in A^*, \langle S_m, M_m\rho M_m^\dagger \rangle \xrightarrow{a}^* \langle S', \rho'\rangle} c(S')\tr(\rho') \\
                         & = \tr(\rho) + \sum_m \ERT'[S_m](\EE_{M_m}(\rho))\\
                         & = \ERT[S](\rho)
        \end{align*}
    \item $S \equiv \while\ M[\oq]=1\ \wdo\ S'\ \wod$: With the induction hypothesis, we only need to show that:
      \begin{equation*}
        \ERT'[S](\rho) = \lim\limits_{k \to \infty} \ERT'[\while^{[k]}[M, S']](\rho).
      \end{equation*}
      Let $A_{S', \downarrow} = \{ a \in A^* \mid \exists \rho'.\ \langle S', \rho \rangle \xrightarrow{a} \langle \downarrow, \rho' \rangle \}$, we have
      \begin{align*}
        & c(T_1(S, a))T_2(S, a, \rho) \neq \zeromat \\ 
        \text{ implies } & a = (\Wh_1 a_1)(\Wh_1 a_2)\cdots(\Wh_1 a_j) e \\
        \text{ where } & a_i \in A_{S', \downarrow} \text{ and } e \text{ is a prefix of } (\Wh_0 \Sk) \mid (\Wh_1 A_{S', \downarrow}) \\
      \end{align*}
      and
      \begin{align*}
         & \text{for } k > 0 \text{, } c(T_1(\while^{[k]}[M, S'], a))T_2(\while^{[k]}[M, S'], a, \rho) \neq \zeromat \\ 
        \text{ implies } & a = (\Br_1 a_1)(\Br_1 a_2)\cdots(\Br_1 a_j)e \\
        \text{ where } & j < k,\ a_i \in A_{S', \downarrow} \text{ and } e \text{ is a prefix of } (\Br_0 \Sk) \mid (\Br_1 A_{S', \downarrow}) \\
      \end{align*}
      Thus, it is sufficient to show that
      \begin{equation*}
        c(T_1(S, a)) = c(T_1(\while^{[k]}[M, S'], a')) \text{ and } T_2(S, a, \rho) = T_2(\while^{[k]}[M, S'], a', \rho)
      \end{equation*}
      for 
      \begin{itemize}
        \item $0 < j \leq k$
        \item $a_i \in A_{S', \downarrow}$ for $0 < i \leq j$
        \item $e$ is a prefix of $(\Wh_0 \Sk) \mid (\Wh_1 a_j)$
        \item $a = (\Wh_1 a_1)(\Wh_1 a_2)\cdots(\Wh_1 a_{j-1}) e$
        \item $e'$ is a prefix of $(\Br_0 \Sk) \mid (\Br_1 a_j)$
        \item $a' = (\Br_1 a_1)(\Br_1 a_2)\cdots(\Br_1 a_{j-1}) e'$.
      \end{itemize}
      This can be check by induction on $j$.
  \end{itemize}
\end{proof}


\section{Proof of Theorem \ref{main_th1}}\label{proof-thm1}

In order to prove this theorem, we need several technical lemmas.

\begin{lemma}
  \label{lemm_term_imp_0}
Let $S$ be a quantum \textbf{while}-program. Then $$\tr (\sem{\while\ M[\oq]=1\ \wdo\ S\ \wod}(\rho)) = \tr \rho$$ implies:\begin{enumerate}\item $\lim\limits_{k \to \infty} \tr (\sem{S} \circ \EE_1)^k(\rho) = 0$; \item $\tr \left( \sem{S} (\EE_1\circ(\sem{S}\circ \EE_1)^k (\rho)) \right) = \tr \left( \EE_1\circ(\sem{S}\circ \EE_1)^k (\rho)\right) $ for all $k \in \NN$.\end{enumerate}
\end{lemma}

\begin{lemma}
  \label{lemm_0_imp_conv}
  Let $\while\ M[\oq]=1\ \wdo\ S\ \wod$ be a quantum loop. Then the following two statements are equivalent: \begin{enumerate}\item $\lim\limits_{k\to\infty} \tr (\sem{S}\circ \EE_1)^k (\rho) = 0$; \item $\sum_{k=0}^{\infty} \tr (\sem{S}\circ \EE_1)^k (\rho)\ {\it converges}.$\end{enumerate}
\end{lemma}

\begin{lemma}
  \label{po_maxp}
  Let $\{ A_n \}$ be an increasing sequence of positive operators on $\HH$, and let $P_1$ and $P_2$ be projections onto subspaces $X_1$ and $X_2$ of $\HH$, respectively. If both $$\bigsqcup_{n=0}^\infty P_1 A_n P_1\ {\rm and}\ \bigsqcup_{n=0}^\infty P_2 A_n P_2$$ exist, then $$\bigsqcup_{n=0}^\infty (P_1 + P_2) A_n (P_1 + P_2)$$ exists, where $P_1 + P_2$ is the projector onto the direct sum space of $X_1$ and $X_2$.
\end{lemma}

The proofs of these lemmas are deferred to Appendix \ref{sec:proofs_of_lemmas}.

To prove Theorem \ref{main_th1}, we have to prove Lemma \ref{split_ert} at the same time.

\begin{lemma}[Lemma \ref{split_ert} and Theorem \ref{main_th1}]
  \label{lemma_6_and_main_th1}
  For any quantum program $S$ and for all $\rho \in W_S$, we have: 
  \begin{align*}
    \tr(\ert[S](A) \cdot \rho) &= \tr(\ert[S] \cdot \rho) + \tr(A \cdot \sem{S}(\rho)) \\
    \text{and} \quad\quad\quad\quad \ERT[S](\rho) &= \tr (\ert[S] \cdot \rho).
  \end{align*}
\end{lemma}

\begin{proof} We proceed by induction on the structure of $S$. 
  Assume that $\rho \in W_S$. 
  \begin{itemize}
    \item $S \equiv \skp$ , $S \equiv \oq := U[\oq]$ and $S \equiv q := \ket{0}$: The result is
      straightforward.
    \item $S \equiv S_1; S_2$: With the induction hypothesis, we have:  
      \begin{align*}
        \tr(\ert[S_1;S_2](A) \cdot \rho)
        & = \tr(\ert[S_1](\ert[S_2](A)) \cdot \rho) \\
        & = \tr(\ert[S_1] \cdot \rho) + \tr(\ert[S_2](A) \cdot \sem{S_1}(\rho)) \\
        & = \tr(\ert[S_1] \cdot \rho) + \tr(\ert[S_2] \cdot \sem{S_1}(\rho)) + \tr(A \cdot \sem{S_2} \circ \sem{S_1}(\rho)) \\
        & = \tr(\ert[S_1](\ert[S_2]) \cdot \rho) + \tr(A \cdot \sem{S_1;S_2}(\rho))\\
        & = \tr(\ert[S_1;S_2] \cdot \rho) + \tr(A \cdot \sem{S_1;S_2}(\rho)).
      \end{align*}
      and
      \begin{align*}
        \ERT[S_1; S_2](\rho) 
        & = \ERT[S_1](\rho) + \ERT[S_2](\sem{S_1}(\rho)) \\
        & = \tr(\ert[S_1]\cdot \rho) + \tr(\ert[S_2] \cdot \sem{S_1}(\rho)) \\
        & = \tr(\ert[S_1](\ert[S_2]) \cdot \rho) \\
        & = \tr(\ert[S_1;S_2] \cdot \rho).
      \end{align*}

    \item $S \equiv \ifb(\guard m\cdot M[\oq] = m \to S_m)\ife$: With the induction hypothesis, we have: 
    \begin{align*}
      \ERT[S](\rho) 
        & = \tr \rho + \sum_{m} \ERT[S_m](\EE_{M_m}(\rho)) \\
        & = \tr \rho + \sum_m \tr(\ert[S_m] \cdot \EE_{M_m}(\rho)) \\
        & = \tr \left( (I + \sum_m \EE_{M_m}^*(\ert[S_m])) \cdot \rho \right)\\
        & = \tr(\ert[S] \cdot \rho).
    \end{align*}
    and
    \begin{align*}
    \tr(\ert[S](A) \cdot \rho) 
        & = \tr \rho + \sum_{m} \tr(\EE_{M_m}^*(\ert[S_m](A)) \cdot \rho) \\
        & = \tr \rho + \sum_{m} \tr(\ert[S_m](A) \cdot \EE_{M_m}(\rho)) \\
        & = \tr \rho + \sum_{m} (\tr(\ert[S_m] \cdot \EE_{M_m}(\rho)) + \tr(A \cdot \sem{S_m}\circ\EE_{M_m}(\rho))) \\
        & = \tr \rho + \sum_{m} \tr(\EE_{M_m}^*(\ert[S_m]) \cdot \rho) + \sum_{m} \tr(A \cdot \sem{S_m}\circ\EE_{M_m}(\rho)) \\
        & = \tr((I + \sum_{m} \EE_{M_m}^*(\ert[S_m])) \cdot \rho) + \tr(A \cdot \sum_{m} \sem{S_m}\circ\EE_{M_m}(\rho)) \\
        & = \tr(\ert[S] \cdot \rho) + \tr(A \cdot \sem{S}(\rho)).
    \end{align*}

    \item $S \equiv \while\ M[\oq]=1\ \wdo\ S'\ \wod$: 
    We first show that $\ERT[S](\rho) = \tr (\ert[S] \cdot \rho)$.
    From $\rho \in W_S$ and Lemma \ref{lemm_term_imp_0},  
      we obtain: 
      $$\EE_1\circ(\sem{S'}\circ \EE_1)^k (\rho) \in W_{S'}$$ 
      for all $k \in \NN$. For simplicity, we denote $\ert[S']$ by $A_{S'}$. Then it follows from Lemma \ref{lemm_ERT_while} that  
      \begin{align*}
    \ERT[\while^{[n]}[M, S']](\rho) 
     &= \sum_{k = 0}^{n-1} \tr \left( (\sem{S'}\circ \EE_1)^{k}(\rho) \right) + \sum_{k = 0}^{n - 1} \ERT[S']\left(\EE_1 \circ (\sem{S'} \circ \EE_1)^{k}(\rho)\right) \\
    & = \sum_{k = 0}^{n-1} \tr \left( (\sem{S'}\circ \EE_1)^{k}(\rho) \right) + \sum_{k = 0}^{n - 1} \tr \left(A_{S'} \cdot \EE_1 \circ (\sem{S'} \circ \EE_1)^{k}(\rho)\right) \\
    & =  \tr \bigg( \big( \sum_{k = 0}^{n-1} (\EE_1^* \circ \sem{S'}^*)^{k}(I) + \sum_{k = 0}^{n - 1} (\EE_1^* \circ \sem{S'}^*)^{k} \circ \EE_1^* (A_{S'}) \big) \cdot \rho \bigg).
      \end{align*}
      By Lemma \ref{ert_while_converge}, we have
      $$ \sum_{k=0}^\infty P_{S} (\EE_1^* \circ \sem{S'}^* )^k(I) P_{S} \text{ and } \sum_{k=0}^\infty P_{S} (\EE_1^* \circ \sem{S'}^* )^k\circ \EE_1^*(A_{S'}) P_{S} $$
      both converge. Then for $\rho \in W_S$, it holds that
      \begin{align*}
        \ERT[S](\rho) & = \lim\limits_{n \to \infty}\tr \bigg( \big( \sum_{k = 0}^{n-1} (\EE_1^* \circ \sem{S'}^*)^{k}(I) + \sum_{k = 0}^{n - 1} (\EE_1^* \circ \sem{S'}^*)^{k} \circ \EE_1^* (A_{S'}) \big) \cdot \rho\bigg) \\
        & =  \lim\limits_{n \to \infty} \tr \bigg( \big( \sum_{k = 0}^{n-1} P_S (\EE_1^* \circ \sem{S'}^*)^{k}(I) P_S + \sum_{k = 0}^{n - 1} P_S (\EE_1^* \circ \sem{S'}^*)^{k} \circ \EE_1^* (A_{S'}) P_S \big) \cdot \rho \bigg) \\
        & =  \tr \bigg( \big( \sum_{k = 0}^{\infty} P_S (\EE_1^* \circ \sem{S'}^*)^{k}(I) P_S + \sum_{k = 0}^{\infty} P_S (\EE_1^* \circ \sem{S'}^*)^{k} \circ \EE_1^* (A_{S'}) P_S \big) \cdot \rho \bigg) \\
        &=  \tr (\ert[S] \cdot \rho).
      \end{align*}
    
    Then we show that $\tr(\ert[S](A) \cdot \rho) = \tr(\ert[S] \cdot \rho) + \tr(A \cdot \sem{S}(\rho))$. Note that since $\rho \in W_S$, we have $P_S \rho P_S = \rho$. Thus, 
        \begin{align*}
& \tr\Big( \big(\sum_{k = 0}^{\infty} P_S (\EE_1^* \circ \sem{S'}^*)^{k} \circ \EE_1^* (\ert[S'](\EE_0^*(A))) P_S\big) \cdot \rho\Big) \\
= & \sum_{k = 0}^{\infty} tr\Big( \big(P_S (\EE_1^* \circ \sem{S'}^*)^{k} \circ \EE_1^* (\ert[S'](\EE_0^*(A))) P_S\big) \cdot \rho \Big) \\
= & \sum_{k = 0}^{\infty} tr\Big( \big((\EE_1^* \circ \sem{S'}^*)^{k} \circ \EE_1^* (\ert[S'](\EE_0^*(A))) \big) \cdot (P_S  \rho P_S) \Big) \\
= & \sum_{k = 0}^{\infty} tr\Big( \ert[S'](\EE_0^*(A)) \cdot \EE_1\circ(\sem{S'}\circ\EE_1)^{k}(\rho)\Big) \\
= &  \tr\Big( \ert[S'](\EE_0^*(A)) \cdot \big( \sum_{k = 0}^{\infty} \EE_1\circ(\sem{S'}\circ\EE_1)^{k}(\rho) \big)\Big) \\
= &  \tr\Big( \ert[S'] \cdot \big( \sum_{k = 0}^{\infty} \EE_1\circ(\sem{S'}\circ\EE_1)^{k}(\rho) \big)\Big) 
+ \tr\Big( \EE_0^*(A) \cdot \big( \sum_{k = 0}^{\infty} \sem{S'}\circ\EE_1\circ(\sem{S'}\circ\EE_1)^{k}(\rho) \big)\Big)\\
= &  \tr\Big( \ert[S'] \cdot \big( \sum_{k = 0}^{\infty} \EE_1\circ(\sem{S'}\circ\EE_1)^{k}(\rho) \big)\Big) 
+ \tr\Big( A \cdot \big( \sum_{k = 0}^{\infty} \EE_0\circ(\sem{S'}\circ\EE_1)^{k+1}(\rho) \big)\Big)\\
= & \tr\Big( \big(\sum_{k = 0}^{\infty} P_S (\EE_1^* \circ \sem{S'}^*)^{k} \circ \EE_1^* (\ert[S']) P_S\big) \cdot \rho\Big) + \tr\Big( \big(\sum_{k = 0}^{\infty} (\EE_1^* \circ \sem{S'}^*)^{k+1}(\EE_0^*(A)) \big) \cdot \rho\Big)
    \end{align*}
    Then we have
    \begin{align*}
        &\phantom{=} \tr(\ert[S](A) \cdot \rho)\\
        & = \tr\Big(\sum_{k = 0}^{\infty} P_S (\EE_1^* \circ \sem{S}^*)^{k}(I) P_S \cdot \rho \Big) + \tr(\EE_0^*(A)\cdot\rho) + \tr\Big(\big(\sum_{k = 0}^{\infty} P_S (\EE_1^* \circ \sem{S}^*)^{k} \circ \EE_1^* (\ert[S](\EE_0^*(A))) P_S\big) \cdot \rho\Big)\\
        & = \tr\Big(\sum_{k = 0}^{\infty} P_S (\EE_1^* \circ \sem{S}^*)^{k}(I) P_S \cdot \rho \Big) + \tr\Big( \big(\sum_{k = 0}^{\infty} P_S (\EE_1^* \circ \sem{S'}^*)^{k} \circ \EE_1^* (\ert[S']) P_S\big) \cdot \rho\Big) \\
        & \quad + \tr\Big(\EE_0^*(A)\cdot(P_S\rho P_S)\Big) + \tr\Big( \big(\sum_{k = 0}^{\infty}  (\EE_1^* \circ \sem{S'}^*)^{k+1}(\EE_0^*(A)) \big) \cdot \rho\Big)\\
        & = \tr\Big(\sum_{k = 0}^{\infty} P_S (\EE_1^* \circ \sem{S}^*)^{k}(I) P_S \cdot \rho \Big) + \tr\Big( \big(\sum_{k = 0}^{\infty} P_S (\EE_1^* \circ \sem{S'}^*)^{k} \circ \EE_1^* (\ert[S']) P_S\big) \cdot \rho\Big) 
        + \tr\Big( \big(\sum_{k = 0}^{\infty} (\EE_1^* \circ \sem{S'}^*)^{k}(\EE_0^*(A)) \big) \cdot \rho\Big) \\
        & = \tr\Big(\sum_{k = 0}^{\infty} P_S (\EE_1^* \circ \sem{S}^*)^{k}(I) P_S \cdot \rho \Big) + \tr\Big( \big(\sum_{k = 0}^{\infty} P_S (\EE_1^* \circ \sem{S'}^*)^{k} \circ \EE_1^* (\ert[S']) P_S\big) \cdot \rho\Big) 
        + \tr\Big( A \cdot \big( \sum_{k = 0}^{\infty} \EE_0\circ(\sem{S'}\circ\EE_1)^{k}(\rho) \big)\Big)\\
        & = \tr(\ert[S] \cdot \rho) + \tr(A \cdot \sem{S}(\rho))
    \end{align*}
    \end{itemize}
\end{proof}

\section{Proof of Theorem \ref{main_th2}}\label{proof-thm2}
\begin{proof} ($\Leftarrow$):
  Since $\HH$ is finite-dimensional, we can conclude directly from Theorem \ref{main_th1}.

($\Rightarrow$):  We proceed by induction on the structure of $S$.
  \begin{itemize}
    \item The case of $S \equiv \skp$, $\oq := U[\oq]$ and $q := \ket{0}$ is straightforward.
    \item $S \equiv S_1; S_2$: From $\ERT[S](\rho) < \infty$, we can assert that $\ERT[S_1](\rho) < \infty$ and $\ERT[S_2](\sem{S_1}(\rho)) < \infty$. By the induction hypothesis we have: 
      \begin{align*}
        \tr(\sem{S_1; S_2}(\rho)) &= \tr(\sem{S_2}(\sem{S_1}(\rho)))\\ &= \tr(\sem{S_1}(\rho)) = \tr \rho.
      \end{align*}
    \item $S \equiv \ifb(\guard m\cdot M[\oq] = m \to S_m)\ife$: It follows from $\ERT[S](\rho) < \infty$ that $\ERT[S_m](\EE_{M_m}(\rho)) < \infty$ for all $m$. By the induction hypothesis we have: $$\tr(\sem{S_m}(\EE_{M_m}(\rho))) = \tr(\EE_{M_m}(\rho))$$ for all $m$. Then we obtain: 
        \begin{align*}
          \tr(\sem{S}(\rho)) &= \sum_m \tr(\sem{S_m}(\EE_{M_m}(\rho)))\\
                             &= \sum_m \tr(\EE_{M_m}(\rho)) = \tr \rho.
        \end{align*}
      \item $S \equiv \while\ M[\oq]=1\ \wdo\ S'\ \wod$: By Lemma \ref{lemm_ERT_while} and $\ERT[S](\rho)$ $< \infty$, we see that $$\lim\limits_{k\to\infty} \tr \left( (\sem{S'} \circ \EE_1)^k(\rho) \right) = 0$$ and $$\ERT[S'](\EE_1\circ(\sem{S} \circ \EE_1)^k(\rho)) < \infty$$ for all $k$. Therefore, we have: $$\tr \left( \sem{S'}(\EE_1\circ(\sem{S'} \circ \EE_1)^k(\rho)) \right) = \tr (\EE_1\circ(\sem{S'} \circ \EE_1)^k(\rho))$$ for all $k$. Consequently, it holds that 
        \begin{align*}
          & \sum_{k = 0}^n \tr (\EE_0 \circ (\sem{S'} \circ \EE_1)^k(\rho)) + \tr((\sem{S'}\circ \EE_1)^{n+1}(\rho)) \\
          = & \sum_{k = 0}^n \tr (\EE_0 \circ (\sem{S'} \circ \EE_1)^k(\rho)) + \tr(\EE_1 \circ (\sem{S'}\circ \EE_1)^{n}(\rho)) \\
          = & \sum_{k = 0}^{n-1} \tr (\EE_0 \circ (\sem{S'} \circ \EE_1)^k(\rho)) + \tr((\sem{S'}\circ \EE_1)^{n}(\rho)) \\
          = & \dots \\
          = & \tr \rho.
        \end{align*}
        Now let $n \to \infty$. We obtain: 
        \begin{align*}
          &\tr(\sem{S'}(\rho)) 
           = \sum_{k = 0}^\infty \tr (\EE_0 \circ (\sem{S'} \circ \EE_1)^k(\rho)) \\
          & = \lim\limits_{n \to \infty} \sum_{k=0}^{n} \tr (\EE_0 \circ (\sem{S'} \circ \EE_1)^k(\rho)) +  \tr((\sem{S'}\circ \EE_1)^{n+1}(\rho)) \\
          & = \tr \rho.
        \end{align*}
  \end{itemize}
\end{proof}

\section{Proofs of of Technical Lemmas}
\label{sec:proofs_of_lemmas}

\subsection{Proof of Lemma \ref{lemm_ERT_while}}
\begin{proof}
  Induction on $n$.
  \begin{itemize}
    \item $n = 0$, $\ERT[\skp](\rho) = 0$.
    \item Suppose 
      \begin{align*}
        \ERT[\while^{[n]}[M, S]]&(\rho) = \sum_{k = 0}^{n-1} \tr \left( (\sem{S}\circ \EE_1)^{k}(\rho) \right) + \sum_{k = 0}^{n - 1} \ERT[S](\EE_1 \circ (\sem{S} \circ \EE_1)^{k}(\rho))
      \end{align*}
      for all $\rho \in \DD(\HH)$. We have
      \begin{align*}
         \ERT[\while^{[n + 1]}[M, S]](\rho) 
        = & \ERT[\ifb\ M[\oq]=0 \to skip\ \guard\ 1 \to S; \while^{[n]}[M, S]\ \ife] \\
          = & \tr \rho + \ERT[S; \while^{[n]}[M, S]](\EE_1(\rho)) \\
          = & \tr \rho + \ERT[S](\EE_1(\rho)) + \ERT[\while^{[n]}[M, S]](\sem{S} \circ \EE_1(\rho)) \\
          = & \sum_{k = 0}^{n} \tr \left( (\sem{S}\circ \EE_1)^{k}(\rho) \right) + \sum_{k = 0}^{n} \ERT[S](\EE_1 \circ (\sem{S} \circ \EE_1)^{k}(\rho)).
      \end{align*}
  \end{itemize}
\end{proof}

\subsection{Proof of Lemma \ref{linearity}}
\begin{proof}
  By induction on the structure of $S$ and using Lemma \ref{lemm_ERT_while} when dealing with a loop. 
\end{proof}

\subsection{Proof of Lemma \ref{lemm_term_space} \cite{Zhou19}}
\label{sub:proof_of_lemmalemm_term_space}
Suppose the Kraus operator-sum representation of $\sem{S}$ is
$$\sem{S}(\rho) = \sum_j E_j \rho E_j^\dagger.$$
Then for any $\ket{\psi} \in \HH_S$, $\tr(\sem{S}(\ket{\psi}\bra{\psi})) = \tr(\ket{\psi}\bra{\psi})$ is equivalent to
\begin{align*}
 \tr(\ket{\psi}\bra{\psi}) - \tr(\sem{S}(\ket{\psi}\bra{\psi})) 
  = & \tr\big(\ket{\psi}\bra{\psi}(I_{\HH_S} - \sum_j E_j^\dagger E_j)\big) \\
  = &\bra{\psi}\big(I_{\HH_S} - \sum_j E_j^\dagger E_j\big)\ket{\psi} = 0
\end{align*}
Since $\sem{S}$ is a super-operator, we have $I_{\HH_S} - \sum_j E_j^\dagger E_j$ is positive. We conclude that $\ket{\psi}$ is in the eigenspace of $I_{\HH_S} - \sum_j E_j^\dagger E_j$ with eigenvalue $0$.

\subsection{Proof of Lemma \ref{ert_while_converge}}
\label{sub:proof_of_lemma_ert_while_converge}
Let $\{ \ket{\psi_j} \}$ be an orthonormal basis of $V_{\while}$ in Lemma \ref{lemm_term_space}. Then by Lemma \ref{lemm_term_imp_0} we have: $$\lim\limits_{k \to \infty} \tr(\sem{S}\circ \EE_1)^k(\ket{\psi_j} \bra{\psi_j}) = 0.$$  By Lemma \ref{lemm_0_imp_conv} we further see that 
     \begin{align*}
      & \lim\limits_{n \to \infty} \bra{\psi_j} \sum_{k=0}^{n} (\EE_1^* \circ \sem{S}^*)^k(I) \ket{\psi_j}= \sum_{k=0}^{\infty} \tr (\sem{S}\circ\EE_1)^k (\ket{\psi_j}\bra{\psi_j})
     \end{align*}
     converges. Therefore, 
     \begin{align*}
       \bigsqcup_{n=0}^\infty (\ket{\psi_j}\bra{\psi_j}) \sum_{k=0}^{n} (\EE_1^* \circ \sem{S}^*)^k(I) (\ket{\psi_j}\bra{\psi_j})
     \end{align*}
     exists. Let $P_{\while}$ be the projector onto $V_{\while}$, then by Lemma \ref{po_maxp} we know that  
     \begin{equation*}
       \bigsqcup_{n=0}^{\infty} P_{\while} \sum_{k=0}^{n} (\EE_1^* \circ \sem{S}^*)^k(I) P_{\while}
     \end{equation*}
     exists. Since $0 \sqsubseteq X \sqsubseteq c I$ for some $c > 0$, we have: \begin{align*}0 &\sqsubseteq (\EE_1^* \circ \sem{S}^*)^k (X)\\ &\sqsubseteq (\EE_1^* \circ \sem{S}^*)^k(cI) = c(\EE_1^* \circ \sem{S}^*)^k(I)\end{align*} for all $k$. Hence,
     \begin{equation*}
       \bigsqcup_{n=0}^{\infty} P_{\while} \sum_{k=0}^{n} (\EE_1^* \circ \sem{S}^*)^k(X) P_{\while}
     \end{equation*}
     also exists. Thus, 
     \begin{align*}
       \sum_{k = 0}^\infty P_{\while} (\EE_1^* \circ \sem{S}^*)^{k}(X) P_{\while}  = \bigsqcup_{n=0}^\infty P_{\while} \sum_{k = 0}^{n} (\EE_1^* \circ \sem{S}^*)^{k}(X) P_{\while} 
     \end{align*}
     converges.


\subsection{Proof of Lemma \ref{lemm_eig_le_1}}
\label{sub:proof_of_lemmas_in_section_sec}
Similar to Lemma 4.1 in \cite{Ying13} (see also Lemma 5.1.10 in \cite{Ying16}).

\subsection{Proof of Lemma \ref{lemm_reduce_jnf}}
\label{sub:proof_of_lemma_lemm_reduce_jnf}
Similar to Lemma 4.2 in \cite{Ying13} (also see \cite{Ying16}, Lemma 5.1.13).

\subsection{Proof of Lemma \ref{lemm_0_imp_conv}}

First, we have: 
\begin{lemma}
  \label{lemm_Jordan_conv}
  Let $J_k(\lambda)$ be a Jordan block of size $k$ corresponding to eigenvalue $\lambda $ with $\abs{\lambda} < 1$. Then $$\sum_{n = 0}^\infty J_k(\lambda)^n$$ converges.
\end{lemma}

\begin{proof}
  We have: 
  \begin{equation*}
    J_k(\lambda)^n =
    \begin{pmatrix}
      \lambda^n & C_n^1 \lambda^{n-1} & C_n^2 \lambda^{n-2} & \dots & C_n^{k-1} \lambda^{n-k+1} \\
      0 & \lambda^n & C_n^1 \lambda^{n-1} & \dots & C_n^{k-2} \lambda^{n-k+2} \\
      \vdots & \vdots & \vdots & \vdots & \vdots \\
      0 & 0 & 0 & \dots & \lambda^n
    \end{pmatrix}.
  \end{equation*}
  Since $\abs{\lambda} < 1$, $$\sum_{n=j+1}^\infty C_n^{n-j}\lambda^{n - j}$$converges for $j = 0, 1, \dots, k - 1$. Thus, $$\sum_{n = 0}^\infty J_k(\lambda)^n$$ converges.
\end{proof}

\begin{proof}[Proof of Lemma \ref{lemm_0_imp_conv}]
  ($\Leftarrow$): Obvious.
  
   ($\Rightarrow$): Since $$(\sem{S}\circ \EE_1)^k (\rho)$$ is positive for all $k$, $$\lim\limits_{k\to\infty} \tr (\sem{S}\circ \EE_1)^k (\rho) = 0$$ implies $$\lim\limits_{k\to\infty} (\sem{S}\circ \EE_1)^k (\rho) = 0.$$ Let $R$ be the matrix representation of $\sem{S} \circ \EE_1$ and $AJ(R)A^{-1}$ be the Jordan decomposition of $R$. Then we have: 
  \begin{equation*}
    \lim\limits_{k\to\infty} R^k(\rho \otimes I)\ket{\Phi} = \lim\limits_{k \to \infty} ((\sem{S} \circ \EE_1)^k(\rho)\otimes I) \ket{\Phi} = 0.
  \end{equation*}
  Since $A$ is nonsingular, we have: 
  \begin{equation*}
    \lim\limits_{k \to \infty} J(R)^k A^{-1} (\rho \otimes I) \ket{\Phi} = 0.
  \end{equation*}
  By Lemma 5.1.10 in \cite{Ying16}, suppose 
  \begin{equation*}
    J(R) = 
    \begin{pmatrix}
      C & 0\\
      0 & D
    \end{pmatrix},\quad 
    S^{-1} (\rho \otimes I) \ket{\Phi} = 
    \begin{pmatrix}
      \ket{u} \\
      \ket{v}
    \end{pmatrix}
  \end{equation*}
  where $C$ is an $r$-dimensional matrix contains Jordan blocks of eigenvalues of modules $1$, $D$ contains Jordan blocks of eigenvalues of module less than $1$, and $\ket{u}$ is an $r$-dimensional vector. 
  It follows from Lemma 5.1.10 in \cite{Ying16} that $C$ is diagonal unitary. Moreover, we have: $$\lim\limits_{k\to\infty} C^k \ket{u} = 0.$$ Then it holds that $\ket{u} = 0$. Since the eigenvalues of all Jordan blocks in $D$ are less than $1$, by Lemma \ref{lemm_Jordan_conv} we know that $\sum_{k=0}^{\infty} D^k$ converges. Hence,  
  \begin{align*}
     \sum_{k = 0}^{\infty} \left( (\sem{S}\circ \EE_1)^k(\rho) \otimes I \right) \ket{\Phi} 
    = & \sum_{k=0}^{\infty} R^{k}(\rho \otimes I)\ket{\Phi} \\
    = & S
    \begin{pmatrix}
      0 \\
      \sum_{k=0}^{\infty} D^k \ket{v}
    \end{pmatrix}
  \end{align*}
  converges. Furthermore,
  \begin{equation*}
    \sum_{k=0}^{\infty} \tr (\sem{S}\circ \EE_1)^k (\rho) = \bra{\Phi } \sum_{k=0}^{\infty} R^{k}(\rho \otimes I)\ket{\Phi}
  \end{equation*}
  converges too.
\end{proof}

\subsection{Proof of Lemma \ref{lemm_term_imp_0}} 
\begin{proof}
  We have:  
  \begin{align*}
     \tr \EE_0 \circ (\sem{S} \circ \EE_1)^k (\rho) + \tr (\sem{S} \circ \EE_1)^{k+1}(\rho) 
    \leq & \tr \EE_0 \circ (\sem{S} \circ \EE_1)^k (\rho) + \tr \EE_1 \circ (\sem{S} \circ \EE_1)^{k}(\rho) \\
    = & \tr (\sem{S} \circ \EE_1)^k (\rho).
  \end{align*}
 Here, the equality holds only if 
  \begin{equation*}
    \tr \left( \sem{S} ((\EE_1 \circ (\sem{S} \circ \EE_1)^k) (\rho)) \right) 
    = \tr (\EE_1 \circ (\sem{S} \circ \EE_1)^k) (\rho)
  \end{equation*}
Moreover, we can derive that
  \begin{align*}
    & \tr (\sum_{k=0}^n \EE_0 \circ (\sem{S} \circ \EE_1)^k(\rho)) + \tr (\sem{S} \circ \EE_1)^{n + 1}(\rho)\\
    \leq & \tr (\sum_{k=0}^{n-1} \EE_0 \circ (\sem{S} \circ \EE_1)^k(\rho)) + \tr (\sem{S} \circ \EE_1)^{n}(\rho) \\
    \leq & \dots \\
    \leq & \tr \rho
  \end{align*}
  by repeatedly using the inequality above. By the assumption we have: 
  \begin{align*}
     \tr (\sem{\while\ M[\oq]=1\ \wdo\ S\ \wod}(\rho)) 
    = & \lim\limits_{n \to \infty} \tr (\sum_{k=0}^n \EE_0 \circ (\sem{S} \circ \EE_1)^k(\rho)) \\
    = & \tr \rho.
  \end{align*}
  Hence, 
  \begin{equation*}
    \lim\limits_{k \to \infty} \tr (\sem{S} \circ \EE_1)^k(\rho) = 0
  \end{equation*}
Further more, for $m > n$ we have: 
  \begin{align*}
    & \tr \left(\sum_{k=0}^m \EE_0 \circ (\sem{S} \circ \EE_1)^k(\rho)\right) + \tr (\sem{S} \circ \EE_1)^{m + 1}(\rho) \\
    & + \tr(\EE_1\circ(\sem{S} \circ \EE_1)^n(\rho)) - \tr \left(\sem{S}(\EE_1\circ(\sem{S} \circ \EE_1)^n(\rho))\right)\\
    \leq & \dots \\
    \leq & \tr \left(\sum_{k=0}^{n} \EE_0 \circ (\sem{S} \circ \EE_1)^k(\rho)\right) + \tr (\sem{S} \circ \EE_1)^{n + 1}(\rho) \\
         & + \tr(\EE_1\circ(\sem{S} \circ \EE_1)^n(\rho)) - \tr \left(\sem{S}(\EE_1\circ(\sem{S} \circ \EE_1)^n(\rho))\right)\\
    = & \tr \left(\sum_{k=0}^{n} \EE_0 \circ (\sem{S} \circ \EE_1)^k(\rho)\right) + \tr(\EE_1\circ(\sem{S} \circ \EE_1)^n(\rho))\\
    = & \tr \left(\sum_{k=0}^{n-1} \EE_0 \circ (\sem{S} \circ \EE_1)^k(\rho)\right) + \tr((\sem{S} \circ \EE_1)^n(\rho))\\
    \leq & \dots \\
    \leq & \tr \rho. 
  \end{align*}

  Let $m \to \infty$. Then we obtain: 
  \begin{equation*}
    \tr(\EE_1\circ(\sem{S} \circ \EE_1)^n(\rho)) - \tr \left(\sem{S}(\EE_1\circ(\sem{S} \circ \EE_1)^n(\rho))\right) = 0.
  \end{equation*}
\end{proof}

\subsection{Proof of Lemma \ref{po_maxp}}

We first have the following: 

\begin{lemma}
  \label{po_bound}
  Let $\{ A_n \}_{n=0}^\infty$ be increasing sequence of positive operators on finite-dimensional Hilbert space $\HH$. If there exists positive operator $A$ in $\HH$ such that $A_n \sqsubseteq A$ for all $n$, then $\bigsqcup_{n} A_n$ exists.
\end{lemma}

\begin{proof}
  Since $\HH$ is finite dimensional, we know that $$\norm{A} = \sup_{\ket{\psi}\neq 0} \norm{A\ket{\psi}} / \norm{\ket{\psi} } $$ exists. 
  Consider increasing sequence $\{ A_n/\norm{A} \}$, we have: $$0 \sqsubseteq A_n/\norm{A} \sqsubseteq A/\norm{A} \sqsubseteq I$$ for all $n$. Then $\{ A_n/\norm{A} \}$ is sequence of quantum predicates. By Lemma 4.1.3 in \cite{Ying16} we know that $$\bigsqcup_{n} A_n/\norm{A}$$ exists. Then $$A' = \norm{A}\bigsqcup_{n} A_n/\norm{A} = \bigsqcup_{n} A_n$$ exists too.
\end{proof}

\begin{proof}[Proof of Lemma \ref{po_maxp}]
  Note that $P_1 + P_2$ is the projector onto the direct sum space $$X_1 \oplus X_2 = \{ \alpha \ket{\psi} + \beta \ket{\varphi}: \ket{\psi} \in X_1, \ket{\varphi} \in X_2, \alpha,\beta\in\CC  \}.$$ We are going to show that for every $j$,  
  \begin{equation*}
    \norm{(P_1 + P_2) A_j (P_1 + P_2)} \leq c (\norm{P_1 A_j P_1} + \norm{P_2 A_j P_2})
  \end{equation*}
  where $c$ only depends on $P_1, P_2$ and does not depend on $j$. 
  
  Let $\{ \ket{\psi_i}  \}$ be an orthonormal basis of $X_1$ and $\{ \ket{\varphi_i}  \}$ an orthonormal basis of $X_2$. 
  Let $\{ \ket{i}\}$ be an orthonormal basis of $X_1 \oplus X_2$. Furthermore, let $n_1 = dim(X_1), n_2 = dim(X_2)$ and $n_3 = dim(X_1 \oplus X_2)$.
By Schmidt orthogonalization there exists matrix $V$ such that
  \begin{equation*}
    \begin{pmatrix}
      \ket{1} \\
      \vdots \\
      \ket{n_3} 
    \end{pmatrix}
    = V 
    \begin{pmatrix}
      \ket{\psi_1} \\
      \vdots \\
      \ket{\psi_{n_1}} \\
      \ket{\varphi_{1}} \\
      \vdots \\
      \ket{\varphi_{n_2}}.
    \end{pmatrix}
  \end{equation*}
  Let $$a = \max \{ \abs{V_{i, j}}: i \leq n_3, j \leq n_1 + n_2 \},$$ and let $\ket{\phi}$ be an arbitrary unit vector in $X_1 \oplus X_2$. Then we can obtain: $$\ket{\phi} = \sum_i \alpha_i \ket{\psi_i} + \sum_i \beta_i \ket{\varphi_i}$$ from $V$, where $\abs{\alpha_i} < a$ and $\abs{\beta_i} < a$ for all $i$.

  Since $\HH$ is finite dimensional, it holds that $$\norm{(P_1 + P_2) A_j (P_1 + P_2)} = \abs{\bra{\phi}(P_1 + P_2) A_j (P_1 + P_2) \ket{\phi}}  $$  for some $\ket{\phi} \in M_1 \oplus M_2$ and $\norm{\ket{\phi} } = 1$. Let $\abs{M_1} = n_1$ and $\abs{M_2} = n_2$. Assume that $$\ket{\phi} = \sum_i \alpha_i \ket{\psi_i} + \sum_i \beta_i \ket{\varphi_i}$$ is obtained from $V$, and denote $(P_1 + P_2) A_j (P_1 + P_2)$ by $A_j'$. Then we have:
  \begin{align*}
    \norm{A_j'} 
    & = \abs{\bra{\phi}A_j' \ket{\phi}}  \\
    & = \Bigg\vert \sum_{i,k=0}^{n_1} \alpha_i^*\alpha_k \bra{\psi_i} A_j' \ket{\psi_k}  
 + \sum_{i,k=0}^{n_2} \beta_i^* \beta_k \bra{\varphi_i} A_j' \ket{\varphi_k} \\ 
    & + \sum_{i=0}^{n_1} \sum_{k=0}^{n_2} \alpha_i^* \beta_k \bra{\psi_i} A_j' \ket{\varphi_k}  
 + \sum_{i=0}^{n_1} \sum_{k=0}^{n_2} \beta_i^* \alpha_k \bra{\varphi_i} A_j' \ket{\psi_k} \Bigg\vert \\
 & \leq a^2 \Bigg(\sum_{i,k=0}^{n_1} \abs{ \bra{\psi_i} A_j' \ket{\psi_k}  }
 + \sum_{i,k=0}^{n_2} \abs{ \bra{\varphi_i} A_j' \ket{\varphi_k}  } \\
 & + \sum_{i=0}^{n_1} \sum_{k=0}^{n_2} \abs{ \bra{\psi_i} A_j' \ket{\varphi_k}  }
 + \sum_{i=0}^{n_1} \sum_{k=0}^{n_2} \abs{ \bra{\varphi_i} A_j' \ket{\psi_k}  }\Bigg).
  \end{align*}

  By Cauchy-Schwarz inequality we can obtain that for a positive operator $A$
  \begin{equation*}
    \abs{\bra{\varphi} A \ket{\psi} } 
    \leq \sqrt{\bra{\varphi}A\ket{\varphi} \bra{\psi} A \ket{\psi} }
    \leq (\bra{\varphi}A\ket{\varphi} + \bra{\psi} A \ket{\psi} ) / 2.
  \end{equation*}
  Furthermore, we have: $$\bra{\psi_i} A_j' \ket{\psi_i}  = \bra{\psi_i} P_1 A_j P_1 \ket{\psi_i}, \qquad \bra{\varphi_i} A_j' \ket{\varphi_i} = \bra{\varphi_i} P_2 A_j P_2 \ket{\varphi_i} .$$ Therefore, it holds that
  \begin{align*}
    \norm{A_j'} 
    & \leq a^2 (n_1^2 + n_2^2 + 2n_1n_2) \\
    & \cdot \left( \max_i\{ \bra{\psi_i} P_1 A_j P_1 \ket{\psi_i} \} 
    + \max_i\{ \bra{\varphi_i} P_2 A_j P_2 \ket{\varphi_i}  \} \right) \\
    & \leq c \left( \norm{P_1 A_j P_1} + \norm{P_2 A_j P_2} \right) .
  \end{align*}

  By the assumption, let $$A^{(1)} = \bigsqcup_j P_1 A_j P_1,\quad A^{(2)} = \bigsqcup_j P_2 A_j P_2.$$  Then we have: 
  \begin{align*}
     (P_1 + P_2) A_j (P_1 + P_2) 
    \sqsubseteq & c(\norm{P_1 A_j P_1} + \norm{P_2 A_j P_2}) I \\
    \sqsubseteq & c(\norm{A^{(1)}} + \norm{A^{(2)}}) I.
  \end{align*}
Finally, by Lemma \ref{po_bound}, we see that $\bigsqcup_j (P_1 + P_2) A_j (P_1 + P_2)$ exists. 
\end{proof}

\subsection{Proof of Lemma \ref{lemm_real_is_enough}}\label{lem_zhou}
 
\begin{proof}
  We have shown in Section \ref{sub:quantum_random_walk} that $Q_n$ 
  and $Q_n'$ are the unique solutions of $$I + E X E^\dagger = X$$ and 
  $$I + E' X' {E'}^\dagger = X',$$ respectively. To obtain $$Q_n = P^\dagger Q_n' P,$$ it is
  sufficient to show that $P^\dagger Q_n' P$ is a solution of 
  $I + E X E^\dagger = X$, which is equivalent to 
  \begin{equation*}
    I + E P^\dagger Q_n' P E^\dagger 
    = P^\dagger Q_n' P 
    = I + P^\dagger E' Q_n' {E'}^\dagger P.
  \end{equation*}
  This can be obtained from $E^\dagger = P^\dagger {E'}^\dagger P$. Using the
  definitions of $E$, $E'$, and $P$, that equation is equivalent to
  \begin{align*}
    & \begin{pmatrix}
      x e^{-i(\beta+\delta)} S_L M_1 & y e^{-i(\beta-\delta)} S_L M_1 \\
      y e^{i(\beta-\delta)} S_L^\dagger M_1 & -x e^{i(\beta+\delta)} S_L^\dagger M_1 \\
    \end{pmatrix} \\
    = 
    & \begin{pmatrix}
      x P_L^\dagger S_L M_1 P_L & y P_L^\dagger S_L M_1 P_R \\
      y P_R^\dagger S_L^\dagger M_1 P_L & -x P_R^\dagger S_L^\dagger M_1 P_R \\
    \end{pmatrix}.
  \end{align*}

  It can be shown by a simple calculation that the corresponding sub-matrices are equal. For example, we have: 
  \begin{align*}
    & P_L^\dagger S_L M_1 P_L \\
    = & (\sum_{k = 0}^{n} e^{i[(\beta + \delta)k + 2 \delta]} \ket{k}\bra{k})
    (\sum_{k = 0}^{n} \ket{k \ominus 1} \bra{k})
     (\sum_{k = 1}^{n - 1} \ket{k} \bra{k})
     (\sum_{k = 0}^{n} e^{-i[(\beta + \delta)k + 2 \delta]} \ket{k}\bra{k}) \\
    = & e^{-i(\beta + \delta)} \sum_{k = 1}^{n - 1} \ket{k \ominus 1}\bra{k} \\
    = & e^{-i(\beta + \delta)}S_LM_1 ,
  \end{align*}
  which implies that $$x e^{-i(\beta + \delta)} S_L M_1 = x P_L^\dagger S_L M_1 P_L.$$ The equalities between other sub-matrices can be proved similarly.
\end{proof}

\section{Calculation of Examples}
\label{sec:calculation_of_examples}

\subsection{Example \ref{ex-rw-EXT}}
\label{sub:example_ex_rw_EXT}
First, we show the following lemma:
\begin{lemma}
  \label{lemm_split_prob}
  Let $S$ be a quantum \textbf{while}-program on $\HH$ such that $\tr(\sem{S}(\rho)) = \rho$ for all $\rho \in \DD(\HH)$, then 
  \begin{equation*}
     \sum_{k = 0}^n \tr (\EE_0 \circ (\sem{S} \circ \EE_1)^k(\rho)) + \tr((\sem{S}\circ \EE_1)^{n+1}(\rho)) = \tr \rho
  \end{equation*}
\end{lemma}

\begin{proof}
  We have:
  \begin{align*}
  & \sum_{k = 0}^n \tr (\EE_0 \circ (\sem{S} \circ \EE_1)^k(\rho)) + \tr((\sem{S}\circ \EE_1)^{n+1}(\rho)) \\
    = & \sum_{k = 0}^n \tr (\EE_0 \circ (\sem{S} \circ \EE_1)^k(\rho)) + \tr(\EE_1 \circ (\sem{S}\circ \EE_1)^{n}(\rho)) \\
    = & \sum_{k = 0}^{n-1} \tr (\EE_0 \circ (\sem{S} \circ \EE_1)^k(\rho)) + \tr((\sem{S}\circ \EE_1)^{n}(\rho)) \\
    = & \dots \\
    = & \tr \rho.
  \end{align*}
\end{proof}

Let $\rho = \ket{L, 1}\bra{L, 1}$. According to Theorem 8 in \cite{Ambainis2001}, we have: 
\begin{equation*}
  \sum_{k = 0}^{\infty} \tr(\EE_0 \circ (\sem{S}\circ \EE_1)^k (\rho)) = \frac{2}{\pi}.
\end{equation*}
Then with Lemma \ref{lemm_ERT_while} and Lemma \ref{lemm_split_prob}, we obtain:
\begin{align*}
  & \phantom{{ }={ }}\ERT[Q_\mathit{qw}](\rho) \\
  & \geq \sum_{k = 0}^{\infty} \tr \left( (\sem{S}\circ \EE_1)^{k}(\rho) \right) \\
  & = \tr(\rho) + \sum_{k = 1}^{\infty} (\tr \rho - \sum_{j = 0}^{k-1}\tr(\EE_0 \circ (\sem{S}\circ \EE_1)^j (\rho))) \\
  & \geq 1 + \sum_{k = 1}^{\infty}(1 - \frac{\pi}{2}) \\
  & = \infty.
\end{align*}

\subsection{Example \ref{ex_2_3_EXT}}
\label{sub:example_ref}
First, we calculate $\sem{Q_1}(\rho_0)$. Let $$S_1 \equiv q := H[q]; p := D[p].$$ Then we have:
\begin{equation*}
  \sem{Q_1}(\rho_0) = \sum_{k = 0}^\infty \EE_{M_0} \circ (\sem{S_1} \circ \EE_{M_1})^k (\rho_0)
\end{equation*}
where $$\EE_{M_0}(\rho) = M_0 \rho M_0^\dagger,\quad \EE_{M_1}(\rho) = M_1 \rho M_1^\dagger.$$ Furthermore, we have: 
\begin{equation*}
  \sem{S_1} \circ \EE_{M_1}(\rho_0) = \ket{-}_q \bra{-} \otimes \ket{2}_p \bra{2}.
\end{equation*}
By induction, we obtain:
\begin{equation*}
  (\sem{S_1} \circ \EE_{M_1})^k(\ket{-}_q \bra{-} \otimes \ket{t}_p \bra{t}) = \frac{1}{2^k} \ket{-}_q \bra{-} \otimes \ket{2^k t}_p \bra{2^k t}
\end{equation*}
for $t > 0$.
Then it follws that
\begin{equation*}
  \sem{Q_1}(\rho_0) = \sum_{k=1}^{\infty} \frac{1}{2^k} \ket{L}_q \bra{L} \otimes \ket{2^k}_p \bra{2^k}.
\end{equation*}

Let $S_2 \equiv p := T_L[p]$. For $Q_2$, we have:
\begin{equation*}
  \ERT[\while^{[k+1]}[N, S_2]](\ket{L}_q \bra{L} \otimes \ket{0}_p \bra{0}) = 1
\end{equation*}
and
\begin{align*}
  &\ERT[\while^{[k+1]}[N, S_2]](\ket{L}_q \bra{L} \otimes \ket{t+1}_p \bra{t+1}) \\
  = &2 + \ERT[\while^{[k]}[N, S_2]](\ket{L}_q \bra{L} \otimes \ket{t}_p \bra{t})
\end{align*}
for $t \geq -1$ and $k \geq 0$. Consequently, we obtain:
\begin{equation*}
  \ERT[\while^{[k]}[N, S_2]](\ket{L}_q \bra{L} \otimes \ket{t}_p \bra{t}) = 2t + 1 
\end{equation*}
for $-1 \leq t < k - 1$.
Furthermore, we have:
\begin{equation*}
  \ERT[Q_2](\ket{L}_q \bra{L} \otimes \ket{t}_p \bra{t}) = 2t + 1
\end{equation*}
for $t \geq -1$ by definition. Finally, we can obtain: 
\begin{equation*}
  \ERT[Q_1; Q_2](\rho_0) = 7 + \sum_{k = 1}^{\infty} \frac{1}{2^k} (2^{k+1} + 1) = \infty
\end{equation*}
by Lemma \ref{linearity}.

\end{appendix}

\end{document}